\documentclass[journal, a4paper, final]{IEEEtran}
\usepackage[english]{babel}
\usepackage{graphicx}
\usepackage{amsmath}
\usepackage{amsfonts}
\usepackage{amssymb}
\usepackage{booktabs}
\usepackage{siunitx}
\usepackage{amsfonts,eucal,bm}
\usepackage{color}
\usepackage{pgfplots}
\usepackage{mleftright} 
\ifCLASSOPTIONcompsoc
\usepackage[caption=false,font=normalsize,labelfont=sf,textfont=sf]{subfig}
\else
\usepackage[caption=false,font=footnotesize]{subfig}
\fi
\pgfplotsset{compat=newest,
	HardPlotStyle/.style={
		xlabel={$p$ of $\mathrm{BSC}$},
		xmax = 0.099,
		xmin = 0.010000,
		ylabel={$\mathrm{WER}$},
		ymax = 1,
		xmajorgrids,
		ymajorgrids,
		grid style=dashed,
		legend pos = south east,
		line width=1pt,
		every axis y label/.style={at={(current axis.north west)},above left=3mm},
	},
	SoftPlotStyle/.style={
		xlabel={$\frac{E_b}{N_0}$ in dB},
		xmax = 3.000000,
		xmin = 1.000000,
		xtick={1, 1.5, 2, 2.5, 3, 3.5, 4, 4.25, 4.5},
		ylabel={$\mathrm{WER}$},
		ymax = 1,
		xmajorgrids,
		ymajorgrids,
		grid style=dashed,
		legend pos = south west,
		line width=1pt,
		every axis y label/.style={at={(current axis.north west)},above left=3mm},
	},
	SizeStyleLarge/.style={
		width = \linewidth,
		label style={font={\footnotesize \color{white!15!black}}},
		tick label style = {font = \footnotesize},
		legend style={font = \footnotesize,legend cell align=left, align=left, draw=white!15!black},
	},
	SizeStyleSmall/.style={
		width = \linewidth,
		label style={font={\footnotesize \color{white!15!black}}},
		tick label style = {font = \footnotesize},
		legend style={font = \footnotesize,legend cell align=left, align=left, draw=white!15!black},
	}
}
\usepackage[]{algorithm2e}

\usepackage[noadjust]{cite}
\makeatletter
\newcommand{\citecomment}[2][]{\citen{#2}#1\citevar}
\newcommand{\citeone}[1]{\citecomment{#1}}
\newcommand{\citetwo}[2][]{\citecomment[,~#1]{#2}}
\newcommand{\citevar}{\@ifnextchar\bgroup{;~\citeone}{\@ifnextchar[{;~\citetwo}{]}}}
\newcommand{\citefirst}{\@ifnextchar\bgroup{\citeone}{\@ifnextchar[{\citetwo}{]}}}
\newcommand{\cites}{[\citefirst}
\makeatother

\pgfplotsset{compat=newest}
\newtheorem{theorem}{Theorem}
\newtheorem{lemma}{Lemma}
\newtheorem{definition}{Definition}

\newtheorem{example}{Example}

\newcommand{\wt}{ \mathrm{\, wt \,} }
\newcommand{\dist}[1]{\ensuremath{\operatorname{dist}(#1)}}
\newcommand{\codename}[1]{\ensuremath{\mathcal{C}_{#1}}}
\newcommand{\comment}[1]{}
\newcommand{\pluseq}{\mathrel{{+}{=}}}
\DeclareMathOperator{\atanh}{atanh} 

\IEEEoverridecommandlockouts

\begin{document}
	
	\title{On Hard and Soft Decision Decoding of BCH Codes}

	\author{
		Martin~Bossert,~\IEEEmembership{Fellow, IEEE}, Rebekka Schulz and~Sebastian Bitzer,~\IEEEmembership{Student Member, IEEE}\\
		\thanks{Manuscript received July 15, 2021; revised January 11, 2022 and April 01, 2022; accepted June 04, 2022.
		}%
		\thanks{
			M. Bossert and R. Schulz are with the Institute of Communications Engineering, Ulm University, Germany (e-mail: martin.bossert@uni-ulm.de, rebekka.schulz@uni-ulm.de).
		}%
	\thanks{
			S. Bitzer is with the Institute for Communications Engineering, Technical University of Munich (TUM), Germany (e-mail: sebastian.bitzer@tum.de). This work was done while S. Bitzer was with Ulm University.
		}
	} 
	
	\maketitle
	
	\begin{abstract}
	The binary primitive BCH codes are cyclic and are constructed by choosing a subset of the cyclotomic cosets. 
    Which subset is chosen determines the dimension, the minimum distance and the weight distribution of the BCH code. 
    We construct possible BCH codes and determine their coderate, true minimum distance and the non-equivalent codes.
    A particular choice of cyclotomic cosets gives BCH codes which are, extended by one bit, equivalent to Reed-Muller codes, which is a known result from the sixties.
    We show that BCH codes have possibly better parameters than Reed-Muller codes, which are related in recent publications to polar codes.  
    We study the decoding performance of these different BCH codes using information set decoding based on minimal weight codewords of the dual code. 
    We show that information set decoding is possible even in case of a channel without reliability information since the decoding algorithm inherently calculates reliability information. 
    Different BCH codes of the same rate are compared and different decoding performances and complexity are observed.
    Some examples of hard decision decoding of BCH codes have the same decoding performance as maximum likelihood decoding. 
    All presented decoding methods can possibly be extended to include reliability information of a Gaussian channel for soft decision decoding.
    We show simulation results for soft decision list information set decoding and compare the performance to other methods.

	\end{abstract}
	
	
	\section{Introduction}
	BCH codes were introduced in \cite{Hoc, BCHa,BCHb} and are a well known 
	class of block codes with various practical applications.
	There exist many hard and soft decision decoding methods for them which can be classified into
	algebraic methods using the Berlekamp-Massey or the Euclidean algorithm, 
	methods based on the dual code, and information set based methods. Many algorithms are already included in textbooks, see e.g. \cite{McWSl, Boss, Blahut}.
	Algebraic decoding works up to half the designed minimum distance which is determined by the construction of BCH codes.
	Interpolation based decoding \cite{guruswami} can be used to extend the decoding radius. 
	Some non-algebraic decoding methods can also decode beyond half the minimum distance.

	Shortly after the introduction of BCH codes, it was shown in 
	\cite{Kolesnik, Kasami} that Reed-Muller (RM) codes are equivalent to
	BCH codes with a particular choice of cyclotomic cosets for their construction, extended
	by a parity bit (see also \cite{McWSl, infoset}). 
	This implies that the hard and soft decision decoding methods for RM codes 
	can also be applied to BCH codes and vice versa. 
	This fact is particularly interesting since there exist many examples for BCH codes that have a larger minimum distance than the RM codes at the same coderate.
	In addition for RM codes the Plotkin construction \cite{Plotkin}
	allows low complexity decoding \cite{Schnabl, Boss}.
	Further, in \cite{RmPolar_Mondelli, RmPolar_Arikan} the relations of 
	polar codes and RM codes are shown. 
	Note that the construction of polar codes was given earlier in \cite{stolte}
	as a modification of the RM construction.
	Therefore, the use of BCH codes has, besides the cyclic property, the advantage 
	of better code parameters than RM and polar codes which possibly leads to a better decoding performance.
	
	In \cite{ferdi}, minimal weight dual codewords are used for decoding beyond half the minimum distance and the decoding methods presented 
	in this work are based on this approach, which was also used in \cite{haeger}.
	We will show that information set decoding of BCH codes without channel reliability
	based on minimal weight dual codewords
	performs in many cases as good as hard decision maximum likelihood decoding.
	The presented decoding method works also for Reed--Solomon and q-ary BCH codes \cite{jiongyue}
	but this will not be discussed here.
	
	For soft decision decoding of BCH codes, three principles which make use of the reliability
	information from the channel, are known. The first \cite{Chase} uses the most unreliable positions
	which contain more likely errors and constructs flip-patterns for these positions. For all
	of these flip-patterns, a classical decoder is applied. 
	The idea is that if a flip-pattern reduces the number of errors to a value less or equal to half the minimum distance,
	it can be decoded by a classical decoder.
	The second \cite{Dorsch} uses the most reliable positions as information set and re-encodes
	with these positions. Also here flip-patterns are used in case the information set contains erroneous
	positions. In \cite{linfos}, both methods are studied under the name ordered statistics.
	The third method uses dual codewords to calculate extrinsic information and reduces
	the number of errors iteratively. 
	Various variants are known and 
	some recent results can be found in \cite{Wu, Khamy, haeger, Liva}. 
	Also, the determination of an information set after applying the extrinsic idea is known \cite{BP_ISD}.
	We will extend the hard decision information set decoding to soft decision decoding
	which uses channel reliability information.
	Since in all studied examples the decoding methods based on unreliable positions
	show worse results than those based on the most reliable positions, 
	we will restrict to the latter case here.

	\paragraph{Contributions}
	In this work, we analyze different BCH codes which have the same rate but various choices of the cyclotomic cosets and we compare their properties.  We recall a known result that every RM code is a BCH code, with a particular choice of the cyclotomic cosets, extended by a parity bit. We show examples where
	BCH codes have better parameters than RM codes.
	The codewords of the dual code can be used to calculate reliability information for each position. The unreliable ones can be used to reduce the number of errors iteratively, whereas the reliable ones can be used to form an information set.
	In \cite{ferdi}, bit flipping based on all minimal weight dual codewords was introduced and used for iterative error reduction.
	For cyclic codes, this algorithm was modified in \cite{SCC2019} using the cyclically different minimal weight codeword polynomials. 
	Then, it was extended in  \cite{arxivboss}, where information set decoding was mentioned but not studied. 
	In \cite{ferdi, SCC2019,arxivboss}, only classical BCH Codes were considered whereas in this work we also consider other BCH codes.
	In addition, here we focus on information set decoding based
	on codewords of the dual code. Further, we introduce the new redundancy set decoding method. 
	Our approach allows information set decoding in case of a binary symmetric channel, which does not provide any reliability
	information. In many cases our approach shows the same decoding performance as
	hard decision maximum likelihood decoding.
	Error reduction decoding based on unreliable positions shows inferior performance compared to the methods based on the most reliable positions. 
	We propose a soft decision decoding algorithm based on hard decision information set decoding.
	We show that including extrinsic reliability information obtained by dual codewords improves the reliability information of the channel considerably.
	Furthermore, we demostrate that our decoding method allows a flexible complexity-performance trade-off.
	Additionally, a novel choice of the flip-patterns is given which is able to predict the decoding performance. 
	Comparisons to selected literature are provided, namely with polar codes \cite{Liva}, multibasis information set decoding \cite{bias_ISD}, and adaptive belief propagation \cite{ABP}.

	\paragraph{Structure}
	In Section~\ref{sec:bch}, we give a short introduction to BCH codes and their dual codes.
	The choice of the cyclotomic cosets is described and illustrated with several examples.
	Then we recall an old result that punctured RM codes are equivalent to BCH codes.
	Finally, we discuss a systematic basis.
	In Section~\ref{sec:decdual}, we describe the main idea of the decoding 
	based on codewords of the dual code.
	We show the inherent reliability contained in the introduced measure
	and relate the measure to classical concepts, namely bit flipping, majority logic decoding, and
	extrinsic information. 
	In the end, we describe three possible decoding concepts. 
	
	The results of the three different decoding methods are discussed in Section~\ref{sec:hardec}.
	The performance of BCH codes with different lengths and rates is presented and analyzed.
	A maximum likelihood bound is introduced in order to compare the obtained word error rates.
	Also, the decoding complexity is addressed.
	
	In Section~\ref{sec:softdec}, we will describe and analyze the
	soft decision extension of information set decoding.
	The results are compared to known results from \cite{Liva,bias_ISD,ABP}.
	Conclusions will be given in the last section and the
	proof of the equivalence of BCH and RM codes, including the permutation, can be found in the appendix.
	
	\section{BCH Codes}\label{sec:bch}
	We restrict ourselves to binary BCH codes and refer to \cite{McWSl, Boss, Blahut}
	for further details.
	Let $\alpha$ be a primitive element of the field $\mathbb F_{2^m}$.
	$\mathbb F_{2^m} [x]$ denotes
	the polynomials $\mod (x^n -1)$, with $n=2^m-1$.
	
	\begin{definition}[Cyclotomic coset] For $n=2^m-1$ the cyclotomic cosets are 
	\begin{equation*}
		K_i = \{i \cdot 2^j\bmod n, j= 0,1, \ldots, m-1\},\  i=0, \ldots, n-1\;.
	\end{equation*}
	\end{definition}
	Two cyclotomic cosets $K_i$ and $K_j$ with $i \neq j$  are either disjoint or identical.
	Their cardinality is $|K_i| \leq m$ and $|K_0| =1$.
	The cyclotomic cosets define the irreducible polynomials $m_{\ell}(x)$  which
	are the product of the linear factors $m_{\ell}(x)=\prod_{j \in K_\ell}(x-\alpha^{-j})$. The irreducible polynomials have coefficients only in $\mathbb F_{2}$.

	\begin{theorem}[BCH Code]\label{th:bch}
		Let $\mathcal M = \cup K_i$ be a union of some cyclotomic cosets for $n=2^m-1$.
		The binary BCH code with generator polynomial
		\begin{equation}
		g(x)= \prod_{j \in \mathcal M} (x-\alpha^{j})	
		\end{equation}
		has length $n=2^m-1$, dimension $k=n-|\mathcal M|$, and designed minimum distance
		$d$ when $\mathcal M$ contains   $d-1$ consecutive integers.
		The true minimum distance $\delta$ can be larger than $d$.
	\end{theorem}
	The generator polynomial $g(x)$ is a product of some irreducible polynomials.
	The BCH$(n,k,d)$ code has the parity check polynomial $h(x)$ with $g(x)h(x)=x^n-1$.
	The dual code of a BCH$(n,k,d)$ 
	code is also a BCH$(n,n-k,d^\perp)$ code with generator polynomial 
	$h(x) =(x^n-1)/g(x)$.
	The codewords of the dual code are $c^\perp(x)= a(x) h(x)$ where $a(x)$ 
	is some information polynomial of degree less than $n-k$.
	As a consequence, the product of any codeword and any dual codeword is zero, i.e.,
	\begin{equation}\label{eq:cwtimesdual}
	c(x) c^\perp(x)= 0\mod (x^n-1) 
	\end{equation}
	which is the main property we will use for decoding.

	\subsection{Choice of Cyclotomic Cosets}
	Algebraic decoding of BCH codes works up to half the minimum designed distance.
	Therefore, given the rate, the cyclotomic cosets have to be chosen such that the designed distance of the BCH code is maximized. For other decoding algorithms 
	different criteria for the selection of the cyclotomic cosets might be applied, 
	for example the weight distribution of the resulting code.
	The presented decoding methods depend on minimum weight codewords of the dual code, 
	hence their weight and their number are influencing the decoding performance and complexity. 
	
	When considering minimal weight dual codewords it may happen that these codewords belong to a subcode of the
	dual code as explained in the following.
	Any dual codeword $b(x)$ can be described as the multiplication of some information 
	polynomial $i(x)$ with the parity check polynomial $h(x)$,
	\begin{equation}
	b(x) = i(x) \cdot h(x)\;.
	\end{equation}
	It is possible that the information polynomial has an irreducible polynomial $m_{\ell}(x)$ as factor which means 
	$m_{\ell}(x) | i(x)$. 
	When this factor does not divide the parity check polynomial,
	$m_{\ell}(x) \, \not| \, h(x)$,
	it can be seen as factor of another parity check polynomial $h_\ell(x)= m_{\ell}(x)\cdot h(x)$ with larger degree,
	\begin{align}\label{eqsubcode}
	b(x) &= i(x) \cdot h(x) = \underbrace{(i_\ell(x) \cdot m_{\ell}(x))}_{=i(x)} \cdot \,h(x) \nonumber\\
	&= i_\ell(x) \cdot \underbrace{(m_{\ell}(x) \cdot h(x))}_{h_\ell(x)}\;.
	\end{align}
	Since $\deg(h_\ell(x)) > \deg(h(x))$, the described code has a smaller dimension, thus, is a subcode. 
	However, the dual code of this subcode is a supercode with generator polynomial 
	$g_\ell(x) = (x^n-1)/h_\ell(x)$  since this supercode has generator polynomial $g_\ell(x) =g(x)/ m_{\ell}(x)$ and thus, a larger dimension and possibly a smaller minimum distance.
	Nevertheless, there exist dual codewords $a(x)$, maybe of larger weight, for which  
	$h(x) | a(x)$ and also $h_\ell(x) \, \not| \, a(x)$ hold.

	In the following, we present several examples for parameters and characteristics of BCH codes for different choices of cyclotomic cosets. 
	Since the presented decoder uses codewords from the dual code the decoding performance will
	depend on the particular choice of cyclotomic cosets.
	Note that the number of possible choices of cyclotomic cosets is exponentially growing with the length $n$. 
	All parameters in the examples are found by computer using SageMath \cite{sage}.
	
	\begin{example}[BCH$(63,31,d)$ Code]\label{ex:ratehalbcode}
		For $n=63$, there exist $13$ cyclotomic cosets $K_i, i \in \{0, 1, 3, 5, 7, 9, 11, 13, 15, 21, 23, \\27, 31\}$. 
		The usual construction of BCH codes to maximize the designed minimum distance would choose the cosets $K_i, i \in \{1,3,5,7,9,11\}$
		and the resulting code is BCH$(63,30,13)$. 
		But fixing the dimension to $31$ with the $13$ cosets, $252$ different BCH codes can be constructed and
		several of them are monomially equivalent, i.e. consist of the same codewords given a suitable permutation of the coordinates. Even if the number and weight of minimum weight dual codewords are identical, differences in the weight distribution of the codes may exist.
		The largest designed minimum distance which was found among them was $d=11$ for
		$K_i, i \in \{1,3,5,7,9,21,27\}$ and we will use this code for bounded minimum distance decoding later
		when comparing the decoding performances.
		For $217$ different codes the minimum weight dual codewords are from a subcode. 	
		We select four codes for further investigation and their parameters, the designed distances  $d$, $d^\perp$,
		the true distances $\delta$, $\delta^\perp$, the number of minimal weight dual codewords $L$, and
		the chosen cyclotomic cosets are listed in Table~\ref{table_chosenCodes_params_63_31}.
		It should be noted that only the cyclically different codewords are counted and therefore, the number of minimal weight codewords is $\leq n L$.
		\begin{table}[h!]
			\centering
            \caption{Parameters of selected codes with $n = 63$ and $k = 31$}\label{table_chosenCodes_params_63_31}
			\begin{tabular}{l|l|l|l|l|l|l}
				$\mathcal{C}$ & $\delta$ & $\delta^\perp$ & $L$ & $d$ & $d^\perp$ & $\{i : K_i \subset \mathcal{M}\}$\\\hline\hline
				$\codename{1}$ & 12 & 10 & 5 & 8 & 6 & $\{5, 9, 11, 13, 21, 23, 27\}$\\
				$\codename{2}$ & 12 & 12 & 35 & 7 & 10 & $\{1, 3, 5, 9, 13, 21, 27\}$\\
				$\codename{3}$ & 12 & 12 & 44 & 7 & 8 & $\{1, 5, 7, 9, 13, 21, 27\}$\\
				$\codename{4}$ & 9 & 12 & 52 & 7 & 12 & $\{11, 13, 15, 21, 23, 31\}$\\
			\end{tabular}
		\end{table}
	\end{example}
	
	\begin{example}[BCH$(63,22,d)$ Code]\label{ex:rmcode}
		Here $168$ BCH codes can be constructed. 
		However, several of them are equivalent. In total only $13$ different codes have minimum weight dual codewords 
		which are not all in a subcode according to (\ref{eqsubcode}).
		Four codes are selected for further investigations and their parameters are listed in Table~\ref{table_chosenCodes_params_63_22}.
		For the first two codes, we had to add $19$, respectively $25$ dual codewords of weight $8$ which
		are not in a subcode.
		Three codes have a different designed and true minimum distance.
		The fourth code has the same designed and true distance and has the most minimum weight dual codewords.
		\begin{table*}[h!]
			\centering
			\caption{Parameters of selected codes with $n = 63$ and $k = 22$}\label{table_chosenCodes_params_63_22}
			\begin{tabular}{l|l|l|l|l|l|l|l|l|l}
				$\mathcal{C}$ & $\delta$ & $\delta^\perp$ & $L_{\delta^\perp}$ & $w_{add}$ & $L_{add}$ & $L$ & $d$ & $d^\perp$ & $\{i : K_i \subset \mathcal{M}\}$\\\hline\hline
				$\codename{1}$ & 16 & 6 & 1 & 8 & 19 & 20 & 11 & 4 & $\{3, 5, 7, 9, 11, 13, 15, 21\}$\\
				$\codename{2}$ & 15 & 6 & 1 & 8 & 25 & 26 & 11 & 6 & $\{1, 3, 5, 7, 9, 13, 21, 23\}$\\
				$\codename{3}$ & 15 & 8 & 30 & & 0 & 30 & 11 & 4 & $\{1, 5, 7, 15, 21, 23, 27, 31\}$\\
				$\codename{4}$ & 15 & 8 & 155 & & 0 & 155 & 15 & 8 & $\{1, 3, 5, 7, 9, 11, 13, 21\}$\\
			\end{tabular}
		\end{table*}
	\end{example}

	The next example shows the case when the minimal weight dual codeword belongs to a subcode. 
	\begin{example}[Dual codewords of BCH$(63,22,d)$]
		The BCH code $\codename{2}$ has the parity check polynomial $h(x) = x^{22} + x^{21} + x^{20} + x^{19} + x^{18} + x^{14} + x^{13} + x^{10} + x^9 + x^7 + x^2 + 1$. 
		The minimum weight dual codewords are $b(x) = x^{56} + x^{51} + x^{23} + x^{17} + x^3 + 1$
		and all cyclic shifts.
		However, $b(x)$ has a factor $m_\ell(x)$ which does not divide  $h(x)$ and is a product of irreducible polynomials since
		the greatest common divisor of $x^n-1$ and $b(x)$ is
		$h_\ell(x) = \gcd(b(x), x^n-1) = x^{31} + x^{27} + x^{25} + x^{23} + x^{21} + x^{19} + x^{17} + x^{15} + x^{13} + x^9 + x^8 + x^5 + x^4 + 1$. Thus, the factor
		$m_\ell(x) = h_\ell(x) / h (x)$ divides $b(x)$ and any cyclic shift
		and thus, $b(x)$ belongs to a subcode.
		Since $h_\ell(x)$ is the parity check polynomial of
		the BCH$(63,31,d)$ code which is a supercode of $\codename{2}$ and
		with $b(x)$ this supercode is checked which has a smaller minimum distance. 
	\end{example}

	The number 127 is prime, hence for this codelength, all corresponding cosets (except for $K_0$) have the same cardinality. This increases the number of possible combinations to achieve a certain dimension, however, it decreases the possible choices of dimensions since $|\mathcal M|$ or $|\mathcal M|-1$ is divisible by $7$.
	
	\begin{example}[BCH$(127,64,d)$ Code]\label{bch127_64}
		The cyclotomic cosets corresponding to $n = 127$ are $K_i, i \in \{0, 1, 3, 5, \\ 7, 9, 11, 13, 15, 19, 21, 23, 27, 29, 31, 43, 47, 55, 63\}$. Since all cyclotomic cosets but $K_0$ contain $7$ numbers, any combination of 9 of these will result in a BCH code with $k = 64$. In total, $\binom{18}{9} = 48\,620$ different codes can be constructed. Here we select four codes to be examined in detail and their parameters (same notation as in the previous examples) are given in Table~\ref{table_chosenCodes_params_127_64}.
		\begin{table}[h!]
			\centering
            \caption{Parameters of selected codes with $n = 127$ and $k = 64$}\label{table_chosenCodes_params_127_64}
			\begin{tabular}{l|l|l|l|l|l|l}
				$\mathcal{C}$ & $\delta$ & $\delta^\perp$ & $L$ & $d$ & $d^\perp$ & $\{i : K_i \subset \mathcal{M}\}$\\\hline\hline
				$\codename{1}$ & 21 & 20 & 28 & 19 & 8 & $\{1, 3, 5, 7, 9, 11, 13, 15, 63\}$\\
				$\codename{2}$ & 20 & 20 & 119 & 13 & 12 & $\{1, 3, 5, 7, 9, 11, 23, 29, 43\}$\\
				$\codename{3}$ & 21 & 22 & 1\,590 & 21 & 8 & $\{1, 3, 5, 7, 9, 11, 13, 15, 19\}$\\
				$\codename{4}$ & 15 & 16 & 651 & 15 & 16 & $\{1, 3, 5, 7, 9, 11, 13, 19, 21\}$\\
			\end{tabular}
		\end{table}
	\end{example}
	
	\subsection{Equivalence of RM and Extended BCH Codes}\label{subsec:equivalence}
	
	Reed-Muller codes are denoted by $\mathcal{R}(r,m)$ where $r$ is the order and $n=2^m$ is the length.
	The dimension is $k=\sum_{i=0}^r {m \choose i}$ and the minimum distance is $d=2^{m-r}$, see \cite{Boss,McWSl}.
	RM codes can be described by various constructions \cites{Schnabl, RmPolar_Arikan, RmPolar_Mondelli}.
	In the following, we will define RM codes as permuted extended BCH codes. 
	A proof can be found in \cite[Ch. 13,\S 5, Th. 11]{McWSl}, where the explicit permutation is not given. 
	However, we give the explicit permutation and a proof that RM codes punctured by one position are cyclic in the appendix. 
	For a BCH code of length $n=2^m-1$ corresponding to a RM code we select a cyclotomic coset if the weight of the binary number $a_0 \ldots a_{m-1}$ representing its index $i$
	is between $0$ and $m-r$.
	Formally, the set $I$ is
	\begin{equation}\label{eq:coset_choice_RM}
	I=\left\{i=\sum_{j=0}^{m-1}a_{j}\cdot 2^j: 0 <\sum_{j=0}^{m-1}a_{j} < m-r, \ a_i \in \{0,1\}\right\}.
	\end{equation}	
	The union of cyclotomic cosets is $\mathcal M = \cup_{i\in I} K_i$. In other words, the factors of $g(x)$
	are the irreducible polynomials $m_i(x), i \in I$.
	\begin{example}[$\mathcal{R}(2,6)$ and  BCH$(63,22,15)$]
		The code $\codename{4}$ from Example~\ref{ex:rmcode} corresponds to the punctured Reed-Muller code $\mathcal{R}(2,6)$ since the binary representation  of the number of the selected cyclotomic cosets $1 \leftrightarrow 00001$,
		$3 \leftrightarrow 00011$,
		$5 \leftrightarrow 00101$,
		$7 \leftrightarrow 00111$,
		$9 \leftrightarrow 01001$,
		$11 \leftrightarrow 01011$,
		$13 \leftrightarrow 01101$,
		$21 \leftrightarrow 10101$
		have all weight larger than $0$ and less than $m-r=4$. 
		The cyclotomic coset
		$15 \leftrightarrow 01111$ was not selected since the weight is $4$.
	\end{example}

	The BCH code  $\codename{4}$ from Example~\ref{bch127_64} is equivalent to the punctured Reed-Muller code $\mathcal{R}(3,7)$. It can be observed that the minimum distance is smaller than those of the other codes.
	
	Let  $(c_{0},\ldots,c_{n-1})$ be a codeword of the BCH code. 
	The permutation of the codeword coordinates to obtain a codeword of the RM code is given by 

		\begin{equation}\label{eq:permutation_func} 
	\pi(i) = \log_\alpha \left( \sum_{j=0}^{m-1}\alpha^{j}\cdot a_j \right)\;,
	\end{equation}
	where $i= \sum_{j=0}^{m-1}2^{j}\cdot a_j,\; i =1 \ldots n, \ a_j \in \mathbb F_2$ and $ \log_\alpha (\alpha^\ell) = \ell$.
	Further, a parity bit $p=\sum_{i=0}^{n-1}c_i$ has to be prepended and the codeword of the $r$-th order Reed-Muller code is
	\begin{equation}\label{eq:extended_cw}
	\left(p,c_{\pi(1)},\ldots,c_{\pi(2^m-1)} \right) \in \mathcal{R}(r,m)\;.
	\end{equation}
	
	\begin{example}[Permutation for $\mathcal{R}(1,3)$ and BCH$(7,4,3)$]
		According to (\ref{eq:coset_choice_RM}), the only chosen cyclotomic coset is $K_1$, since
		$I=\left\{2 i_1+i_0: 0 < i_0+i_1 < 2\right\}=\left\{1,2\right\}$ 		
		and $K_1$ includes $1$ and $2$.
		Using $p(x)=x^3+x+1$ as primitive polynomial of $\mathbb F_{2^3}$, the 
		generator polynomial of the BCH code is $g(x) = x^3+x+1$.
		With the same $p(x)$  the permutating function according to (\ref{eq:permutation_func}) 		can be calculated, for example consider
		\begin{align*}
		\pi(5)&=\pi(1\cdot2^2+0\cdot2^1+1\cdot2^0)\\
		&=\log_\alpha(1\alpha^2+0\alpha^1+1\alpha^0)\\
		&=\log_\alpha(\alpha^6)=6\;.
		\end{align*}		
		Using (\ref{eq:extended_cw}), the codewords of $\mathcal{R}(1,3)$ are given by 
		$\left(p,c_{0},c_{1},c_{3},c_{2},c_{6},c_{4},c_{5}\right)$ where		
		$p=\sum_{i=0}^{6}c_i$. 
	\end{example}

	In general, the special choice of cyclotomic cosets is suboptimal with respect to maximizing the designed distance.
	Therefore, the minimum distance of a RM code is often smaller than 
	the true minimum distance of a comparable extended BCH code.
	An example for this statement are the BCH codes $\codename{3}(127,64,21)$ and $\codename{4}(127,64,15)$ from Example~\ref{bch127_64}, where the latter corresponds to an RM code.
	This observation is especially interesting due to the close relation of RM and polar codes \cites{RmPolar_Arikan, RmPolar_Mondelli}.

	\subsection{Calculation of A Systematic Basis}\label{subsec:sytematic}
	Let $g(x)$ be the generator polynomial of a BCH$(n,k,d)$ code. The degree of $g(x)$ is $n-k$.
	One possible generator matrix $\mathbf G$ consists of the $k$ rows $x^j g(x), j =0, \ldots, k-1$
	and this is a $k \times n$ matrix.
	Let $ S \subset \{0,1, \ldots, n-1 \}$ be a set of $k$ distinct positions.
	If the submatrix $\mathbf G_S$, which consists of the columns $\ell \in S $, has full rank
	$k$ then $S$ is called information set and these positions uniquely determine the remaining $n-k$ positions of any codeword. Since  $\mathbf G_S$ has full rank, by 
	Gaussian elimination of  $\mathbf G$ it is possible to calculate the $k \times n$  matrix $\mathbf G_I$
	which has the unity matrix at the positions $S$.
	It is known that for BCH codes not any $k$ positions are an information set. 
	Which $k$ positions form an information set is a code property and not a property of the chosen generator matrix.
	
	If the submatrix $\mathbf G_S$ does not have full rank, the linear dependent
	columns must be substituted by other columns of $\mathbf G$ which are linearly independent.
	Since $\mathbf G$ has rank $k$, a $\mathbf G_S$ with full rank can always be found.
	
	Cyclic codes have the property that any $k$ consecutive positions $ S_c$ are an information set.
	Let $ S_c = \{k_0,k_0+1, \ldots, k_0+k-1 \mod n \}$. 
	These position can be cyclically shifted by $x^j, j=n-k-k_0$ such that
	$ S_s = \{n-k, \ldots, n-1 \}$. Encoding of the information symbols
	$c_I(x)= c_{n-k}x^{n-k} + \ldots + c_{n-1}x^{n-1}$  is done by calculating the remainder $c_R(x)$  
	when dividing $c_I(x)$ by
	the generator polynomial $g(x)$ denoted by $\mathrm{rem} (g(x),c_I(x))$. The codeword is then 
	$c(x)=c_I(x)+c_R(x)$.
	Obviously, the remainder has degree $< n-k$ and its $n-k$ coefficients are the redundancy 
	part of the codeword. If necessary, one can shift back $c(x)$ such that the information
	positions are at  $\{k_0,k_0+1, \ldots, k_0+k-1 \mod n \}$ again.
	
	In case of an information set of consecutive positions we can calculate the generator matrix $\mathbf G_C$
	which consists of the $k$ rows 
	\begin{equation}
	c_\ell (x) +x^\ell, \ \ell= n-k, \ldots, n-1\;,
	\end{equation}
	where $c_\ell (x)=\mathrm{rem}(g(x),x^\ell)$ is the remainder of the division $x^\ell :g(x)$.
	This matrix $\mathbf G_C$ has the form 
	\begin{equation}\label{eq:sysmatrix}
	\mathbf G_C = (\mathbf G_R | \mathbf I_k)\;,
	\end{equation}
	where $\mathbf I_k$ is the $k \times k$ identity matrix.

	\section{Decoding Based on the Dual Code}\label{sec:decdual}
	
	In case of a binary symmetric channel (BSC), we receive $r(x)=c(x)+e(x)$ where $e(x)=x^{e_0} +  x^{e_1} + \ldots + x^{e_{\tau-1}} $ is the error polynomial of weight $\tau$
	and $c(x)$ is the transmitted codeword. We can calculate
	\begin{align}\label{eq:rectimesdual}
	r(x) c^\perp(x)&=   (c(x)+e(x)) c^\perp(x) \nonumber\\
	&=  e(x) c^\perp(x)\mod (x^n-1) \;,
	\end{align}
	where we have used (\ref{eq:cwtimesdual}). Since the dual code is cyclic, we can choose a dual codeword $c^\perp(x)=b(x)$ of minimal weight $\delta^\perp$
	of the form $b(x)=1 + x^{b_1} + \ldots + x^{b_{\delta^\perp-1}}$.
	We calculate $w(x)=r(x) b(x)\mod (x^n-1)$ 
	which is equal to 
	$w(x)= e(x) b(x)\mod (x^n-1)$, according to (\ref{eq:rectimesdual}). The polynomial $w(x)$  can be written as sum of cyclic shifts of $e(x)$, namely
	\[	
	\begin{array}{rcl}
	w(x)
	&=& x^{e_0} +  x^{e_1} + \ldots + x^{e_{\tau-1}} +\\
	& & x^{e_0+ b_1} +  x^{e_1 + b_1} + \ldots + x^{e_{\tau-1} +b_1} +\\
	&  & \vdots\\
	& & x^{e_0+ b_{\delta^\perp-1}} +  x^{e_1 + b_{\delta^\perp-1}} + \ldots + x^{e_{\tau-1} +b_{\delta^\perp-1}}\;, 
	\end{array}
	\]
	where the exponents $ e_i + b_j$ are calculated$\mod n$. 
	Every coefficient of $w(x)$ equal $1$ is either an error  $x^{e_i}$ or a shifted
	error  $x^{e_i+b_j}$. The $1$ could also come from an addition of an error and an even
	number of shifted errors. 
	Similarly, for coefficients equal $0$ the position is error free and also no shifted error at this position or an error and an odd number of shifted errors are added, or an even 
	number of shifted errors is added.

	The main idea \cite{arxivboss} is to shift back  
	the shifted errors to their original positions using
	$x^{-b_i} w(x) \mod (x^n-1), i=0, \ldots, \delta^{\perp}-1$.   
	Recall that $x^{-b_0} w(x) =w(x)$. 
	Then it is counted how often position $j$ is 
	equal to $1$ 
	in all $\delta^\perp$ shifts of $w(x)$. This number is called $\Phi_j$ and calculated by
	\begin{equation}
	{ \Phi_j} =  \sum\limits_{i \in \mathrm{supp} \, b(x)} \  w_{(j+i)\bmod n},\ j=0, 1, \ldots, n-1\;. 
	\end{equation}
	Obviously, using a polynomial $b'(x)$ which is a cyclic shift of $b(x)$ where the
	zero coefficient is $1$ 
	would result in the same $\Phi_j$. Therefore, we can use cyclically different polynomials only.
	In case of $L$ cyclically different minimal weight polynomials $b^{(\ell)}(x), \ell=0, \ldots, L-1$
	we can do this counting for each of these and add them.
	The calculation of $\Phi_j$ using only the coefficients of $r(x)$ is  
	\begin{equation}\label{eq:phi_coeff}
	\Phi_j= \sum\limits_{\ell=0}^{L-1} \sum\limits_{i=0}^{\delta^\perp-1}  \left(\sum\limits_{l=0}^{\delta^\perp-1}
	r_{(j+ b^{(\ell)}_i - b^{(\ell)}_l)\bmod n}\bmod 2 \right).
	\end{equation}
	
	\subsection{Inherent Reliability Information}
    In order to explain intuitively that $\Phi_j$ contains information on the reliability of position $j$, we calculate the expected value of the weight $\omega$ of $w(x)$ and follow the ideas in \cite{ferdi}. 
    Interested readers find a more detailed stochastic analysis in \cite{yuan21}.
	The coefficient $w_i$ is $1$ if an odd number of the $\tau$ error positions
	overlap with the $\delta^\perp$ nonzero positions of $b(x)$.
	However, we only can calculate how often this happens when considering all possible
	errors of weight $\tau$.
	This number $W \in \mathbb N$ depends only on $ \tau, \delta^\perp$, and $n$ and is
	\begin{equation}
	W=  { \delta^\perp \choose 1} {n- \delta^\perp \choose \tau- 1} +{ \delta^\perp \choose 3} {n- \delta^\perp \choose \tau- 3} + \ldots +q\;,
	\end{equation}
	where 
	\[
	q= \left\{ \begin{array}{ll}
	{ \delta^\perp \choose \tau} {n- \delta^\perp \choose 0}, & \tau \ \mathrm{odd}\\
	{ \delta^\perp \choose \tau-1} {n- \delta^\perp \choose 1}, & \tau \ \mathrm{even}\;.\end{array}
	\right.
	\]
	Thus, the expected contribution of one error of weight $\tau$ to one
	coefficient $w_i$ is
	$W/ {n \choose \tau}$. Therefore, the expected
	weight $\omega$ of $w(x)$ is
	\begin{equation}\label{DA}
	E[\omega] =  \frac{n W}{{n \choose \tau}}\;.
	\end{equation}
	Now, we follow the ideas in \cite{arxivboss} in order to calculate the expected values for 
	$\Phi_j$ in case of an error and a non-error position, denoted by $E[\Phi_e(\tau)]$ and  $E[\Phi_c(\tau)]$, respectively. 
	At an error position $j$,
	the frequency of occurrence $\Phi_j$
	has the expected  value
	\begin{equation}\label{eq:expected_phierror}
	E[\Phi_e(\tau)] =  \frac{E[\omega]}{\tau} L\;.
	\end{equation}
	Since the expected weight $\omega$  is for  $\tau$ errors 
	for each 
	polynomial $b^{(\ell)}(x)$, we have to divide by $\tau$ and multiply by $L$
	because we have $L$ different polynomials.
	For a non-error position the expected value of  $\Phi_j$
	is
	\begin{equation}\label{eq:expected_phicorrect}
	E[\Phi_c(\tau)] =   \frac{E[\omega] (\delta^\perp-1)}{n-\tau}   L\;.
	\end{equation}
	An error position contributes to $\delta^\perp-1$ positions of the $n-\tau$ non-error positions. Thus,
	we multiply the expected weight by  $\delta^\perp-1$ and divide by the number of non-error positions.
	Again, we have $L$ polynomials.
	\begin{figure}[h]
		\centering
		\subfloat[Expected and average values of $\Phi$ dependent on $\tau$
		\label{fig_ValsOfPhi_Hard}]
		{\includegraphics[]{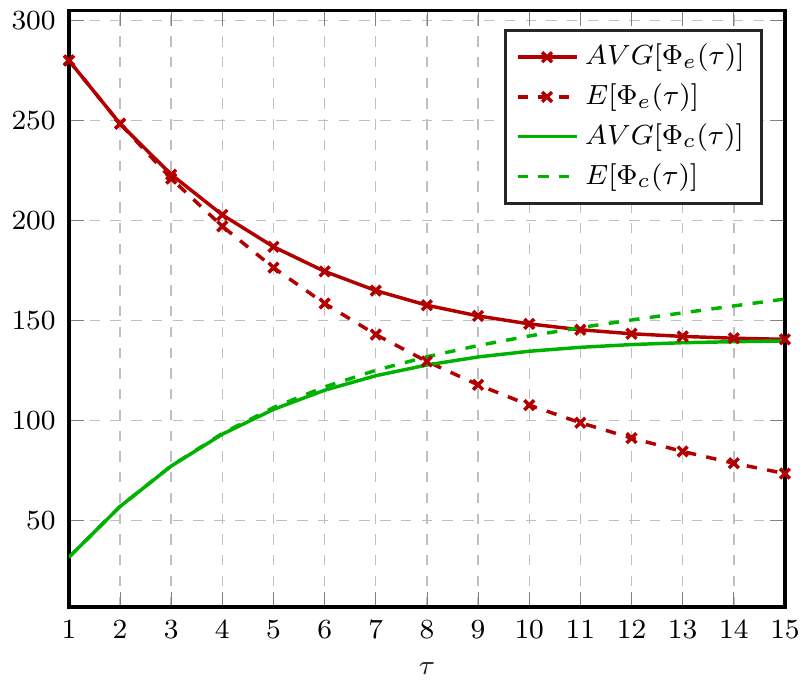}}
		\hfil
		\subfloat[Statistics of non-error and error positions
		\label{fig_ErrorStatistics_63-24}]
		{\includegraphics{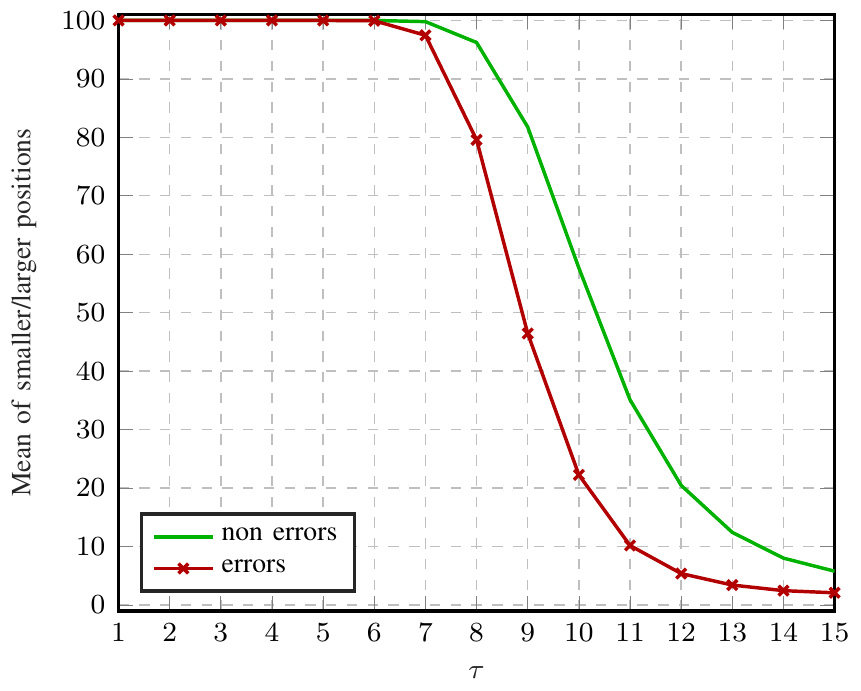}}
		\caption{Aspects of $\Phi$ for BCH$(63,24,15)$}
	\end{figure}
	
	\begin{example}[Comparison of  Expected Values and Simulated Averages]
		We use the $35$ different dual codewords of minimum weight $\delta^\perp = 8$ of BCH$(63,24,15)$. The simulation uses 20\,000 random errors of weight $\tau$ and calculates the corresponding values of $\Phi$. The calculated values together with simulation results are shown in Fig.~\ref{fig_ValsOfPhi_Hard}. The mean of all occurred values for error positions is $AVG[\Phi_e(\tau)]$. The expected weight $E[\Phi_e(\tau)]$ is calculated according to (\ref{eq:expected_phierror}). Respectively, $AVG[\Phi_c(\tau)]$ gives the average values of $\Phi(\tau)$ for correct positions, whereas $E[\Phi_c(\tau)]$ is predicted by (\ref{eq:expected_phicorrect}) and gives the expected value for $\Phi(\tau)$ for non-error positions. For $\tau < 6$ the smallest occurred value for an error position is larger than the largest value for a correct position. In this case, erroneous and correct values are clearly separated.
		It can be observed that the estimated values for non-error positions are 
		closer to the simulated averages than those for the
		error positions. This is due to the above-mentioned effect that a $w_j=1$, 
		which corresponds to a shifted error at one position can, in addition, be 
		shifted to other error positions.
		
	\end{example}
	
	The values of $\Phi_j$ contain inherent reliability information 
	because the smaller values are more likely error free and the larger values are 
	more likely erroneous. 
	This implies decoding concepts for different
	decoding strategies based on the unreliable positions 
	or based on the reliable positions or based on both.

	\begin{example}[Distribution of Error and Non-Error Values of $\Phi$]\label{ex:err_distr}
		According to the description in Section~\ref{sec:decdual}, higher reliability is indicated by a smaller value of $\Phi$ of the corresponding position.
		A smaller reliability corresponds to a larger value. 
		For a decoding concept based on reliable or based on unreliable positions,
		it is interesting how many error positions have a smaller reliability than all correct positions or how many correct positions have a larger reliability than all error positions.
		Fig.~\ref{fig_ErrorStatistics_63-24} shows the answer to this question
		for the same code BCH$(63,24,15)$ that was used before. 
		For each error weight, we simulated $20\,000$ errors and examined the values of $\Phi$ corresponding to error and non-error positions. The red curve corresponds to the average percentage of error positions that show a smaller reliability than all correct positions. Respectively, the average percentage of correct positions with larger reliability than all errors is depicted in green. For $\tau < 7$ both curves yield values of $100\%$, thus in these cases the $\tau$ error positions correspond to the $\tau$ largest values of $\Phi$. With increasing error weight the values of both curves decrease. Still, the non-error positions correspond on average in $20-30\%$ more cases to the largest reliability values than it is the case for error positions and the smallest reliabilities.
	\end{example}

	Before we describe the decoding concepts we will show that the measure $\Phi_j$ 
	can be related to known decoding methods.
	
	\subsection{Relation to Known Decoding Concepts}
	There exist two classical decoding methods based on codewords of the dual code.  
	The bit flipping decoding \cite{Gallager}, introduced for low-density 
	parity-check codes
	and the majority logic decoding  \cite{Massey} for arbitrary codes.
	Further, the concept of extrinsic information is also based
	on codewords of the dual code.
	For short, we will call them BiFl, MaLo, and ExIn, respectively.
	We will follow the derivation in \cite{arxivboss}. 
	Given the dual codeword $b(x)=1 + x^{b_1} + \ldots + x^{b_{\delta^\perp-1}}$ the vector  $\mathbf p$ with support $\{n-1,n-1-b_1, \ldots, n-1-b_{\delta^\perp-1}\}$ is a codeword of the dual code. The reverse operation
	relates the polynomial multiplication to the scalar product (compare (\ref{eq:phi_coeff})). 
	The scalar product of $\mathbf p$ with the received vector $\mathbf r$ is called a parity check  
	$\langle \mathbf p, \mathbf r \rangle =  r_{n-1}+ r_{n-1-b_1}+ \ldots
	+ r_{n-1-b_{\delta^\perp-1}}$
	and results in $0$ or $1$.
	Each coefficient $w_i$ of $w(x)=r(x) b(x)$ corresponds to the scalar product of the 
	shifted dual codeword $\mathbf p$ with $\mathbf r$, 
	for example $w_{n-1}= r_{n-1}+ r_{n-1-b_1}+ \ldots + r_{n-1-b_{\delta^{\perp-1}}}$.
	By cyclic shifting the  $L$ codewords, we can create  $n L$
 minimum weight codewords, and thus,  $\mathbf p_t, t =0, \ldots, nL-1$ parity checks. 
	Clearly, the $nL$ coefficients of $w^{(\ell)}(x), \ell =0, \ldots, L-1$ are the results of the 
	parity checks $\langle \mathbf p_t, \mathbf r \rangle, t =0, \ldots, nL-1$.
	
	\paragraph{Relation to MaLo}
	In MaLo decoding, each position $j$ is decoded separately
	using a set of parity checks ($\mathbf p_i, i=1, \ldots, J$) with 
	the following property \cite{Boss}:
	at position $j$ all checks are $1$ and at
	all other positions  only one of the checks has a $1$. 
	In other words $\cap_i \mathrm{supp}(\mathbf p_i )=j$.
	If the majority of the checks is 
	not fulfilled ($=1$), position $j$ is decided to be erroneous. 
	Note that for $\mathrm{wt} (\mathbf p_i)= \delta^\perp$ it follows that the number $J$ of checks which can exist, is bounded by $J\leq \lfloor (n-1)/(\delta^\perp-1) \rfloor$.
	If we restrict the calculation of $\Phi_j$ by using only those polynomials
	$b^{\ell}(x)$ with $b^{\ell}_j=1$ and all the other coefficients are only $1$
	in one of the polynomials, this is identical
	to MaLo decoding. 
	Therefore, the calculation of $\Phi_j$ is an extension of MaLo by using more checks for 
	position $j$ and violating the restriction that the other positions
	should have ones at disjoint positions.
	Since we use all cyclic shifts of the $L$ polynomials $b^{\ell}(x)$
	for any position $j$ there are  $\delta^\perp L$ checks  with a
	$1$ at this position. However, the other positions are 
	$1$ in more than one check, since $\delta^\perp L> \lfloor (n-1)/(\delta^\perp-1) \rfloor$. 
	The $\Phi_j$ represents the voting if position $j$ is an error or not.
	Using more checks for this voting about position $j$ turns out as an advantage.

	\paragraph{Relation to BiFl}
	In BiFl position $j$ in the received vector $\mathbf r$ is flipped, resulting in the vector 
	$\mathbf r_j$ and 
	the scalar products $\langle \mathbf p_t, \mathbf r \rangle$  are compared
	with $\langle \mathbf p_t, \mathbf r_j \rangle$. The position $j$
	for which the number of scalar products with result
	$1$ is reduced most is considered as error.
	Flipping position $j$ is the addition of $x^j b^{(\ell)}(x) \mod (x^n-1)$ to $w^{(\ell)}(x)$. 
	Note that this is equivalent to the correlation
	between $x^j b^{(\ell)}(x) \mod (x^n-1)$ and $w^{(\ell)}(x)$. 
	In order to measure the change in the scalar products for BiFl we define the value $\Delta_j$ by
	\begin{equation}\label{delta}
	\Delta_j = \sum\limits_{\ell=1}^L \mathrm{wt} (w^{(\ell)}(x) + x^j b^{(\ell)}(x)) - \mathrm{wt} (w^{(\ell)}(x))\;.
	\end{equation}
	\begin{lemma}[\cite{arxivboss} Relation between $\Phi$ and $\Delta$]\label{bfmalo}
		\[L \delta^\perp - 2\Phi_j= \Delta_j\;.\]
	\end{lemma}
	\begin{proof}
		We use only one polynomial $w(x) = e(x) b(x)\mod (x^n-1)$.
		The value  of $\Delta_j$ only depends on the non-zero positions of 
		$x^j b(x)$ which are $\mathcal J =\{j, j+b_1, j+ b_2, \ldots, j+ b_{\delta^\perp-1}\}$. 
		If the value of $w_i,\  i \in \mathcal J$ is $1$, $-1$ is added to $\Delta_j$ and if the value is $0$,
		$1$ is added. 
		$\Phi_j$ is given by $\Phi_j= \sum_{j \in \mathcal J} w_j$ , where $w(x)$ is shifted by $0,-b_1, \ldots, -b_{\delta^\perp-1}$
		and corresponds to the number of ones in these $\delta^\perp$ positions.
		The number of zeros is then $\delta^\perp -\Phi_j$ and 
		$\Delta_j$ is the number of zeros minus the number of ones, 
		thus, $\Delta_j =  (\delta^\perp -\Phi_j) - \Phi_j $.
		This is valid for each of the $L$ polynomials, which completes the proof.
	\end{proof}
	
	The concept of low-density parity-check codes is to create
	minimal weight dual codewords. However, only $n-k$ codewords are used as checks
	to get a parity check matrix where $n-k$ is in general much smaller than $nL$.
	
	\paragraph{Relation to ExIn}
	Assume the check $r_j + r_{t_1} + \ldots +r_{t_{\delta^\perp-1}}=s_j \in \{0,1\}$.
	With this check equation we can calculate the information of $\delta^\perp-1$ positions
	about the remaining position. 
	This information of the positions $t_1, \ldots, t_{\delta^\perp-1}$ about position $j$ is 
	called extrinsic information. Note that the name extrinsic information is also used
	in a different meaning assuming statistical independence.
	Here we use extrinsic information to describe the collection
	of the information of various combinations of other positions about one position. This exploits the weak law of large numbers
	rather than statistical independence.
	$\Phi_j$ can be interpreted as the extrinsic information for position $j$.  
	The calculation of $\Phi_j$ uses $L \delta^\perp$ checks, therefore  $n(L \delta^\perp)$ checks
	are used for all $n$ positions. There are $n L$ checks for BiFl, however, for each check, $\delta^\perp$ extrinsic informations are calculated which results in  $\delta^\perp(n L)$ checks
	which is the same number.
	Using the extrinsic information to update the reliabilities of positions is known as belief propagation which can be done in the same way with the $\Phi_j$.

	\subsection{Decoding Concepts}
	At the receiver, we calculate $\Phi_j, j=0, \ldots, n-1$ and sort 
	the positions according to the value  $\Phi_j$,  which corresponds to their reliability.
	The order $\Phi_{i_0} \leq \Phi_{i_1} \leq \ldots \leq \Phi_{i_{n-1}}$
	assigns the index set $ S =\{ i_0, i_1, \ldots, i_{n-1}\}$.
	The most reliable position is $i_{0}$ and the most unreliable is  $i_{n-1}$.
	Here we restrict to a BSC and hard decision decoding.
	Based on these reliabilities we can use the unreliable positions and reduce the number of errors iteratively by flipping these positions. 
	However, we also can use the most reliable positions as an information set and it is also possible to use 
	both, reliable and unreliable positions.

	The first concept is called {\em error reduction decoding} and is an iterative method. In the first step the 
	most unreliable position is flipped. 
	Note that the values  $\Phi_j$ have not to be sorted here but only the position $max$ with the maximum value $\Phi_{max}$ must be determined.
	After flipping this position in the received polynomial the values $\Phi_j$ are recalculated.
	Then, again, the most unreliable position is flipped and the $\Phi_j$ are recalculated and so on.
	The decoding is either stopped when a valid codeword is found or when a maximum number of iterations is reached. In the latter case, a decoding failure is declared.
	
	In {\em information set decoding} the positions $i_0, \ldots, i_{k-1}$
	are used to calculate a systematic basis according Section~\ref{subsec:sytematic}.
	If these positions contain linear dependent columns the positions 
	$i_k, i_{k+1}, \ldots$ are used until a systematic basis 
	$\mathcal B = \{\ell_0, \ell_1, \ldots, \ell_{k-1}\}$ is found.
	Then encoding with the received bits $r_{\ell_0}, r_{\ell_1}, \ldots, r_{\ell_{k-1}}$ 
	as information is done. Here a decoding failure is not possible.
	This can be extended to a list decoder by flipping each information bit $r_{\ell_i}, i=0, \ldots, k-1$ and encoding again. The list size is then $k$ and the codeword which has the smallest Hamming distance to the received polynomial is chosen.
	In case of several codewords with the same distance, one of them is chosen at random.
	The list size could possibly be enlarged by including flips of the unreliable positions of the basis.
	Note that the inherent reliability of the 
	$\Phi_j$ makes information set decoding possible.
	Without this reliability, only a trial and error search for an error free systematic basis could be done.
	
	The third decoding concept is the  
	{\em redundancy set decoding}.
	It uses the fact that for cyclic codes any $k$ coherent positions are systematic positions
	as shown in Section~\ref{subsec:sytematic}. 
	We will use the positions $n-k, \ldots, n-1$ 
	as systematic positions in the following.
	Let $\overline r (x) = c(x) + e(x)$ be received. The addition of a codeword $c_w(x)$ 
	to  $\overline r (x) $ does not change the error positions.
	If we choose the information part $c_I(x)= \overline r_{n-k} x^{n-k}+ \ldots + \overline r_{n-1} x^{n-1}$
	and the redundancy part $c_R(x) = \mathrm{rem}(g(x), c_I(x))$
	then $r(x) = \overline r (x) + c_w(x) = r_0 + r_1 x + \ldots + r_{n-k-1} x^{n-k-1}$.
	The $k$ positions $r_{n-k} = r_{n-k+1} = \ldots =r_{n-1}=0$ are zero, 
	however, some of these positions may be erroneous. We will illustrate all 
	steps of this new decoding concept with examples.
	\begin{example}[BCH$(15,7,5)$]
		The generator polynomial is $g(x)=x^8 + x^7 + x^6 + x^4 + 1 $ 
		and the parity check polynomial is $h(x)= x^7 + x^6 + x^4 + 1$. 
		The dual code has only one cyclically different minimal weight polynomial $b(x)= x^{11} + x^3 + x^2 + 1$. 
		In fact we could use also other cyclic shifts of this polynomial
		for which the coefficient at $x^0$ is $1$. 
		In particular $x^4 b(x) = h(x) \mod (x^n-1)$.
		Let the codeword $c(x)= x^{14} + x^{12} + x^{11} + x^{10} + x^{9} + x^{6} + x^{4} + x^{3} + x $ 
		be transmitted and the error be 
		$e(x)= x^{14} + x^2+1$
		which has a weight beyond half the minimum distance.
		Then the received polynomial is 
		$\overline{r}(x) = x^{12} + x^{11} + x^{10} + x^{9} + x^{6} + x^{4} + x^{3} + x^{2} + x + 1$
		From the systematic part of the received polynomial $x^{12} + x^{11} + x^{10} + x^{9}$
		we calculate the codeword
		$c_w(x)=x^{12} + x^{11} + x^{10} + x^9 + x^7 + x^5 + x^4 + x$ 
		and then 
		$r(x)= x^7 + x^6 + x^5 + x^3 + x^2 + 1$.
	\end{example}

	Clearly, the $\Phi_j$ are identical when calculated with $\overline r (x) $ or $ r (x) $ 
	according to (\ref{eq:phi_coeff}) since they only depend on the error and the dual codewords.
	The value $\Phi_j$ corresponds to the reliability of position $j$.
	Let $\mathcal B = \{\ell_0, \ell_1, \ldots, \ell_{k-1}\}$ be the index set of the 
	sorted reliabilities of the positions $n-k, n-k+1, \ldots, n-1$ and $\ell_0$ is the
	most unreliable one which corresponds to the largest value $\Phi_j$. 
	The remaining $n-k$ positions $0,1, \ldots, n-k-1$ are also
	sorted as the index set  $\mathcal G = \{j_0, j_1, \ldots, j_{n-k-1}\}$ and 
	$j_0$ is the most reliable symbol. Note that the sets  $\mathcal B$ and $\mathcal G$
	are sorted in opposite order.
	The errors in the systematic part are more likely in the first positions of the set $\mathcal B$ whereas the first part of $\mathcal G$ is likely error free. 
	We take $\mu$ unreliable positions from the systematic 
	part and $\mu$ reliable from the redundancy part. 
	The choice of $\mu$ is influencing the decoding performance
	and here we choose $\mu \approx k/2$.
	The idea is that the unreliable positions of
	the systematic part contain all errors which are in this systematic part whereas the reliable positions of the
	redundancy part are error free.
	\begin{example}[$\Phi$ for BCH$(15,7,5)$]
		We calculate $w(x)=r(x) b(x)=  x^{14} + x^{13} + x^{11} + x^{10} + x^{5} + x^{4} + x^{3} + x^{2} + x + 1$.
		In addition, we need the three cyclic shifts by $-2,-3$ and $-11$ of $w(x)$. Now, we count the ones at the positions and 
		get the values
		
		\[\Phi= 
		\left(4,\,3,\,4,\,3,\,2,\,2,\,1,\,2,\,3,\,2,\,2,\,3,\,2,\,3,\,4\right)\;.
		\]
		The index sets are $\mathcal B =\{6, 0, 3, 5, 1, 2, 4\}$ which means the positions
		$\{14, 8, 11, 13, 9, 10, 12\}$, 
		and $\mathcal G= \{6, 4, 5, 7, 1, 3, 0, 2\}$.
		We choose $\mu=3$ and verify that the three positions of $\mathcal B$, $\{6, 0, 3\}$ contain 
		the error in the systematic part in position $14$.
		Further, the three positions of $\mathcal G$, $ \{6,4,5\}$ of the redundancy part contain no errors.
	\end{example}

	The matrix $\mathbf G_R$ according to (\ref{eq:sysmatrix}) has $k$ rows and $n-k$ columns. 
	From this matrix we extract a $\mu \times \mu$ square matrix $\mathbf D$
	by taking $\mu$ rows  $\ell_0, \ell_1, \ldots,\ell_{\mu-1}$ according
	the index set  $\mathcal B$ and  $\mu$ columns $j_0, \ldots,  j_{\mu-1}$
	according the index set  $\mathcal G$.
	Then from the received values a vector $\mathbf r_G=(r_{ j_{n-k-1}}, \ldots, r_{ j_{n-k-\mu}})$ is determined.
	Assuming that $\mathbf r_G$ is error free as well as the positions $\ell_\mu, \ell_{\mu+1}, \ldots,\ell_{k-1}$
	and $\mathbf D$ is invertible, then the vector $\mathbf \varepsilon = \mathbf r_G \mathbf D^{-1}$ 
	are the errors in the systematic part. 
	
	We consider a simple case, where only one error is in $\ell_0$ and all other systematic positions $\ell_1, \ldots,\ell_{k -1}$ are error free.
	Then the redundancy part $\mathrm{rem}(g(x), x^{\ell_0})$ should be identical to the redundancy 
	part $\mathbf r_G$ at those positions of $j_0, \ldots ,j_{n-k -1}$ which are error free.
	Because of linearity, this holds for any number of errors as long as all errors are in the 
	$\mu$ of the $k$ systematic positions.
	
	\begin{example}[Matrix for Decoding BCH$(15,7,5)$]
		The matrix $\mathbf G_R$ 
		\begin{align*}
		\mathbf G_R&=\left(\begin{array}{rrrrrrrr}
		1 & 0 & 0 & 0 & 1 & 0 & 1 & 1 \\
		1 & 1 & 0 & 0 & 1 & 1 & 1 & 0 \\
		0 & 1 & 1 & 0 & 0 & 1 & 1 & 1 \\
		1 & 0 & 1 & 1 & 1 & 0 & 0 & 0 \\
		0 & 1 & 0 & 1 & 1 & 1 & 0 & 0 \\
		0 & 0 & 1 & 0 & 1 & 1 & 1 & 0 \\
		0 & 0 & 0 & 1 & 0 & 1 & 1 & 1
		\end{array}\right)\\
		\end{align*}
		from which we select the three rows  $ \mathcal B = \{6,0,3\}$ and three columns
		$\mathcal G = \{6,4,5\}$ gives the matrix $\mathbf D$ and its inverse
		\begin{align*}
		\mathbf D=
		\left(\begin{array}{rrr}
		1 & 0 & 1 \\
		1 & 1 & 0 \\
		0 & 1 & 0
		\end{array}\right)	\;, \;
		\mathbf D^{-1}=
		\left(\begin{array}{rrr}
		0 & 1 & 1 \\
		0 & 0 & 1 \\
		1 & 1 & 1
		\end{array}\right)	\;.	
		\end{align*}
		The vector $\mathbf r_G=  \left(1,\,0,\,1\right) $
		and the error in the systematic part is $\mathbf e=(0,0,0,0,0,0,1)$.
		With this the decoded codeword is
		$\hat c(x) = c_w(x) + \mathrm{rem}(x^{14}, g(x)) + x^{14}$
		which is the transmitted one.
	\end{example} 

	Two extensions of this 
	decoding concept to list decoding are possible.
	The first one is to
	use cyclic shifts of the received polynomial $\overline r$.
	With these shifts
	any $k$ consecutive positions can be shifted to the consecutive systematic positions
	$n-k, \ldots, n-1$. With this shifted received polynomial $\overline r_s$
	the same decoding can be done. The $\Phi_j$ does not need to be calculated again 
	but only shifted accordingly. 
	However, a different codeword $c_w(x)$  has to be calculated and added
	and the sorting of the new sets $\mathcal B$ and $\mathcal G$ must be done again.   
	With these shifts, the Dirichlet principle (also known as pigeonhole principle or Dirichlet's box principle) is used implicitely. If $\tau$ 
	random errors are in $n$ positions, 
	the principle says that when splitting the $n$ positions into $\nu$ disjoint parts 
	there is at least one part which contains $\leq \tau/\nu$ errors.
	But, there is also at least one part which contains $\geq \tau/\nu$ errors. Therefore,
	we can shift the different parts to the fixed systematic positions and have different error
	constellations in the systematic part and in the redundancy part.
	
	The second possible extension to list decoding is the flipping  of 
	the used positions of the redundancy part which is the vector $\mathbf r_G$.
	For any of these vectors with flipped positions, 
	as in information set decoding, the multiplication
	$\mathbf \varepsilon = \mathbf r_G \mathbf D^{-1}$ must be calculated.
	
	\section{Hard Decision Decoding}\label{sec:hardec}
	In this section, we describe three different decoding algorithms
	and analyze their complexity.
	Then we show the decoding performance for selected BCH codes of
	different lengths and rates based on simulations.
Since for a BSC the probability that $\tau$ errors in $n$ positions occur can be calculated, we simulate random errors of weight $\tau$ and run the decoding algorithm. 
    Let $p_\tau$ be the probability that $\tau$ errors can not be decoded by the used decoding algorithm, where $p_\tau$ is determined by the simulations.
	Then the  word error rate (WER) can be calculated for a BSC channel
	with error probability $p$,
	using the 
	binomial distribution by 
	\begin{equation}
	\mathrm{WER}(p) =\sum\limits_{\tau=1}^{n} p_\tau {n \choose \tau} p^\tau (1-p)^{n-\tau}\;.
	\end{equation}
	The simulation results will be discussed and comparisons will be made.
	But first, we derive a maximum-likelihood (ML) bound, which is used 
	to relate the achieved performance to ML decoding.
	
	For the remainder of this section, a BSC with error probability $p$ is assumed.
	All the algorithms have been implemented in SageMath \cite{sage}.  
	
	\subsection{Maximum Likelihood Bound}
	Since ML decoding is not feasible for many codes, we use a lower bound
	from \cite[Th. 9.28, p. 384]{Boss}.
	The idea is based on the fact that in simulations we know the transmitted 
	codeword $c(x)$ as well as the error $e(x)$.
	In addition we have the decision $\hat c (x)$ of the used hard decision decoding algorithm
	when $r(x)=c(x) + e(x)$ was received. 
	An ML decoder decides for 
	the codeword with smallest Hamming distance $c_m(x) = \arg \min_c \dist{r,c}$.
	If $\dist{r, \hat c} \geq \dist{r, c}$ we assume that the ML decoder 
	decodes  correctly $c_m(x)=c(x)$, which is not guaranteed.
	If $\dist{r, \hat c} < \dist{r, c}$, we assume that the ML decoder makes a decoding error because it will find the same codeword or another one
	with the same or smaller distance.
	In addition if the outcome of the decoding algorithm is a failure and the decoder gives no estimation, we assume that the ML decoder would decode correctly.
	The calculated WER is therefore a lower bound and a ML decoder can not perform better.
	If the performance of the decoding algorithm is identical to the bound,
	it has the same performance as ML decoding. 
	
	This bound can be tightened in case of a 
	list-decoder when the output consists of multiple codewords possibly with equal distance to the received sequence.
	Let $\tau = \dist{r, c}$ be the true number of errors and $\hat{\tau}$ the minimum estimated number of errors.
	Furthermore, let the list $\mathcal{L}$ of codewords found by the decoder only contain codewords with corresponding error weight $\hat{\tau}$.
	Then, the number of error events which are relevant for the ML lower bound can be updated based on this decoding result by
	\begin{equation}
	\text{\#errors} = \begin{cases}
	0                                           &\text{if } \hat{\tau} > \tau \\
	1-\frac{1}{\left|\mathcal{L}\right|}   = \frac{\left|\mathcal{L}\right|-1}{\left|\mathcal{L}\right|}    &\text{if } \hat{\tau} = \tau \text{ and } c\in \mathcal{L} \\
	1-\frac{1}{\left|\mathcal{L}\right|+1} = \frac{\left|\mathcal{L}\right|}{\left|\mathcal{L}\right|+1}    &\text{if } \hat{\tau} = \tau \text{ and } c\notin \mathcal{L} \\
	1                                           &\text{if } \hat{\tau} < \tau\;. 
	\end{cases}
	\end{equation}
	The update rule for the cases with $\tau = \hat{\tau}$ is based on the following considerations.
	Let $s$ be the number of codewords with distance $\tau$ to the received sequence.
	Then, an ML decoder will choose a wrong codeword with probability $1-\frac{1}{s}$.
	While $s$ is unknown, it can be lower bounded by $\left|\mathcal{L}\right|$, i.e., the number of found codewords, which also lower bounds the error probability of the ML decoder.
	In the case that $c$ is not in the list, an even tighter bound can be obtained, as it is known that at least $\left|\mathcal{L}\right|+1$ codewords with distance $\tau$ exist.
	
	\subsection{Decoding Algorithms}
	For all decoding algorithms the reliability information $\Phi$ according to (\ref{eq:phi_coeff})
	has to be calculated.
	According to \cite{SCC2019,arxivboss}, this can be done by Algorithm~\ref{calcphi} where $L$ dual minimal weight codewords are used which are cyclically different and have a one at position $x^0$.
	
	\medskip
	\begin{algorithm}[h]
		\DontPrintSemicolon
		\SetKwInput{Input}{input}\SetKwInput{Output}{output} 
		\Input{$r$, cyclically different minimum weight checks $\mathcal{W}$}
		\Output{$\Phi$}
		\For{$b \in \mathcal{W}$}{
			$w = r \cdot b \mod x^n-1$\;
			\For{cyclic shifts $w_{shifted}$ of $w$}{
				\For{$exp \in \mathrm{exponents}(w_{shifted})$}{
					$\Phi_{exp} \pluseq 1$	
				}
			}
		}
		
		\caption{Calculation of $\Phi$}\label{calcphi}
	\end{algorithm}
	
	\medskip
	For binary codes, the polynomial multiplication is the XOR of $\delta^\perp$ cyclic shifts of the
	received vector. The complexity to calculate $\Phi$ is $L$ polynomial multiplications and $L \delta^\perp$ cyclic shifts. Then these $L \delta^\perp$ binary vectors are added.
	Each coordinate of $\Phi$ is an integer between $0$ and $L \delta^\perp$.

	\subsubsection{Error Reduction Decoding}
	The error reduction is an iterative decoding method, where in each iteration
	the unreliable positions according to $\Phi$ are flipped to calculate a new received vector,
	see Algorithm~\ref{eredec}.
	With this new received vector, $\Phi$ is recalculated which can be done by modifying
	$w^{(\ell)}(x)$ as follows.
	Let $j_1, \ldots, j_\mu$ be the flipped positions, 
	then $w^{(\ell)}_{new}(x)= w^{(\ell)}(x) + x^{j_1} b^{(\ell)}(x) + \ldots + x^{j_\mu} b^{(\ell)}(x)$.
	Therefore, the complexity is the number of iterations times the calculation of $\Phi$.
	In addition, a sorting operation in each iteration has to find the $\mu$ largest values
	of $\Phi$ which takes not more than  $\mu$ times $n$ comparisons.
	There exist several variants, which and how many unreliable positions are flipped in each iteration.
	Further, a list decoding variant is given as follows: if more than $maxflip$ positions possess the maximum value of $\Phi$, multiple subsets of these values can be chosen. For each of these subsets, an iterative decoder can be started.
	
	\subsubsection{Information Set Decoding}
	
	For this decoding, the values of $\Phi$ have to be calculated only once and a sorting has to be done to find the $k$ most reliable positions. There exist several possibilities to find the  $k$ most reliable positions which are a systematic. Here we sort the columns of the generator matrix according to their reliability and choose the linearly independent pivot columns as information set.
	Therefore, the complexity is the calculation of $\Phi$, the sorting, and the pivot calculation.
	Since we have a systematic encoding matrix we can do list decoding
	by flipping some of the $k$ systematic positions. The flip-pattern
	is a binary vector that contains ones at the positions which are flipped. We choose the list size $1+k + k(k-1)/2$ which means we try all flip-patterns of weight one and two, corresponding to single and double errors.
In general, the decoding performance improves when using more flip-patterns of larger weight, however, the performance can not be better than ML decoding.	
	
	\medskip
	\begin{algorithm}[h]
		\DontPrintSemicolon
		\SetKwInput{Input}{input}\SetKwInput{Output}{output}
		\Input{$r$, minimum weight checks $\mathcal{W}$}
		\Output{$\hat{c}$}
		
		$\hat{c} = r$\;
		
		\For{$\mathrm{iter}<\mathrm{max iter}$}{
			\If{$\hat{c} \in \mathcal{C}$}{
				break\;
			}
			calculate $\Phi$\;
			
			$J = \{j \in \{0,1,\dots, n-1\}: \Phi_j = \max(\Phi)\}$\;
			\If{$|J| > maxflip$}{
				$J = $ choose $maxflip$ random positions of $J$
			}
			flip $\hat{c}_j \quad \forall j \in J$
		}
		\caption{Error Reduction}\label{eredec}
	\end{algorithm}
	
	\medskip
	\begin{algorithm}[h]
		\DontPrintSemicolon
		\SetKwInput{Input}{input}\SetKwInput{Output}{output}
		\Input{$r$, $\Phi$, generator matrix $G$}
		\Output{$\hat{c}$}
		
		\tcp{find Information set from most reliable positions according to $\Phi$}
		$G_{sorted} = $sort columns of $G$ according to reliability of $\Phi$\;
		$I = $ pivot columns of $G_{sorted}$
		
		\tcp{apply flip-patterns}
		\For{$\forall e \in \mathbb{F}_2^k, 0 \leq \wt(e)\leq2$}{
			$r_I = r|_I + e$\;
			$\hat{c} = r_I \cdot G_I^{-1} \cdot G$\;
		}
		\Return $\hat{c} = \mathrm{argmin}(\mathrm{dist}(\hat{c}, r))$
		
		\caption{Information set decoding}\label{Alg-ISD}
	\end{algorithm}
	
	\medskip
	\subsubsection{Redundancy Set Decoding}
	This decoding algorithm 
	makes use of the fact that for cyclic codes $k$ consecutive positions
	are an information set. Because the code is cyclic we can fix the systematic positions
	and use a precalculated systematic encoding matrix. In order to exploit the Dirichlet principle,
	we can use cyclic shifts of the received vector. The values of $\Phi$ do not have to be recalculated but can also be shifted. Only the sorting in the systematic and in the redundancy part has to be done again.
	Since this decoding was described in detail with examples in the last section we omit a formal description here.

	\subsection{Decoding Performance}
	We start with simulation results of the four BCH$(63,31)$ codes from Example~\ref{ex:ratehalbcode} for the transmission over a BSC channel. 
	Because a BSC channel does not provide reliability information, the selection of the information set is based on the inherent reliability of the $\Phi$ values only. Additionally, all flip-patterns of weight one and two are applied, which corresponds to 497 flip-patterns. As shown in Fig.~\ref{fig_Hard_63_31}, the results of information set decoding (ISD) are equal to the ML bound (ML-LB) for all considered codes, which means that ISD has the same performance as ML decoding. The performance of the first three codes, which share the same true distance of $\delta = 12$, is the same. Since the first code \codename{1} only uses 5 cyclically different dual codewords, the decoding complexity is reduced compared to the other codes with 35 and 44 dual codewords, respectively. Code \codename{4} has a smaller true minimum distance $\delta = 9$ and, thus, shows a small loss in performance for smaller channel error probabilities, even though more dual codewords are used for the calculation of the reliability. In contrast, for error reduction decoding (ERD) \codename{4} shows a better performance than the other codes, whose results have been omitted in Fig.~\ref{fig_Hard_63_31}. For all codes, ERD shows a worse performance than ISD. For comparison, the curve for bounded minimum distance decoding of a BCH code with designed minimum distance $d=11$ is included which is described in Example~\ref{ex:ratehalbcode}.

	\begin{figure}[h]
		\centering	
		\subfloat[Hard decision decoding of BCH$(63,31)$. 	\label{fig_Hard_63_31}]
		{\includegraphics[]{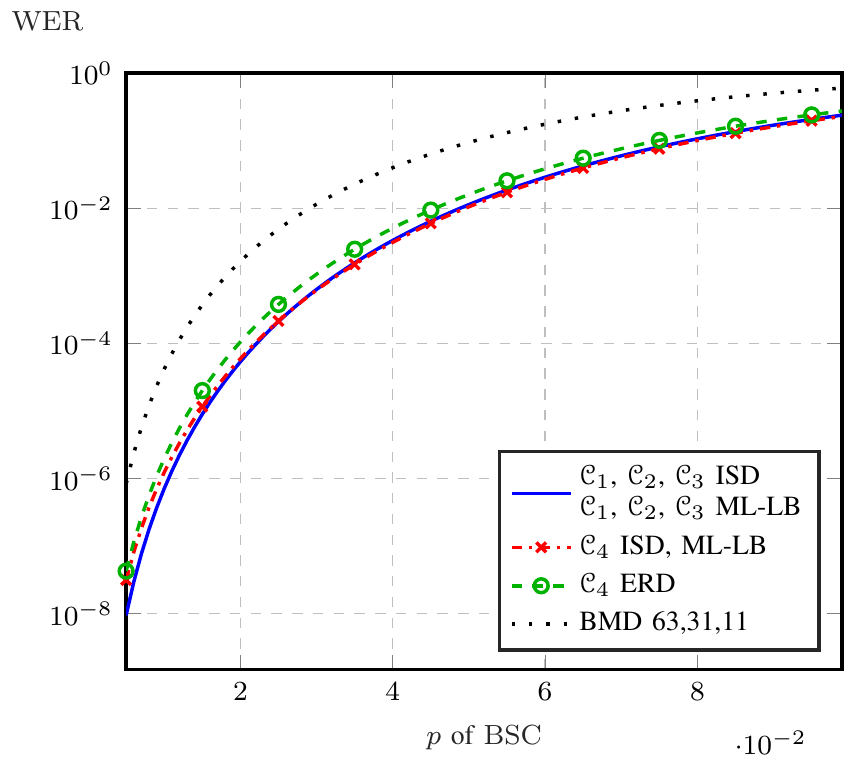}}
		\hfil
		\subfloat[Hard decision decoding of BCH$(63,22)$. 	\label{fig_Hard_63_22}]
		{\includegraphics[]{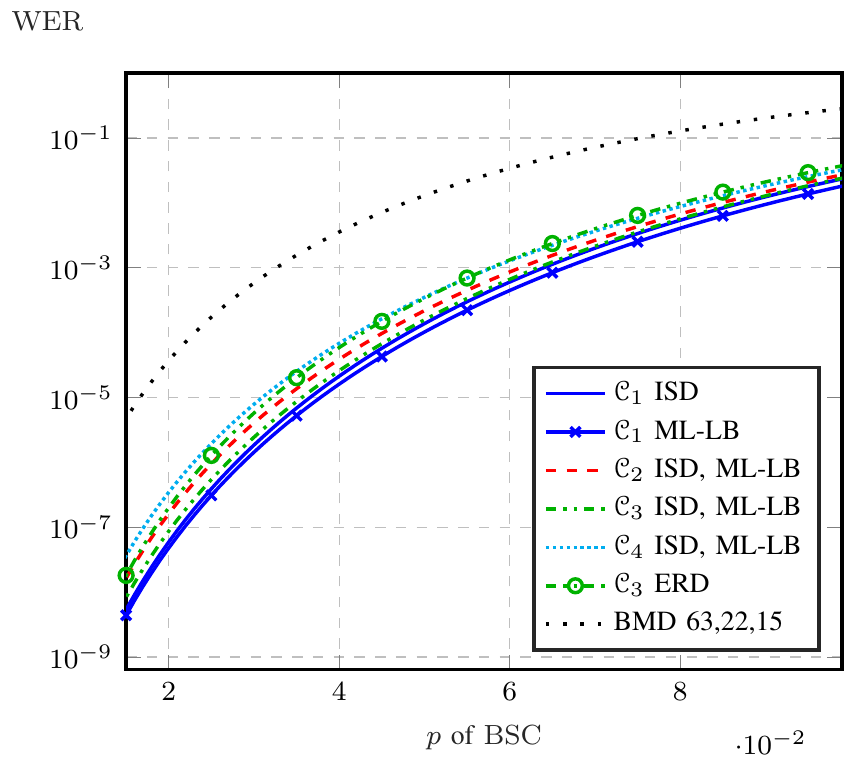}}
		
		\caption{Hard decision decoding of multiple codes with codelength 63.}
	\end{figure}
	
	The simulation results of the four BCH$(63,22)$ codes from Example~\ref{ex:rmcode}
	are visualized in Fig.~\ref{fig_Hard_63_22}. Again all flip-patterns of weight one and two are used for ISD, which for this dimension corresponds to 254 flip-patterns.
	The best performance is achieved by ISD of \codename{1}, which has the largest true minimum distance, $\delta = 16$. Since this code uses only 20 dual codewords, the complexity for the calculation of $\Phi$ is reduced compared to the other codes. Additionally to the single minimum weight dual codeword of weight $6$, $19$ dual codewords of weight $8$ are added. The second best performance has \codename{3}, which uses $30$ dual codewords. For \codename{2}, additionally to the single dual codeword of weight $6$, the $25$ dual codewords of weight $8$ are used for decoding. This code shows a worse performance than \codename{3}. \codename{4} corresponds to the punctured RM code and has $155$ minimum weight dual codewords. This code shows the worst performance of the considered codes. This is an example for which the decoding of a BCH code performs better than the one for the RM code. The performance of \codename{2}, \codename{3}, and \codename{4} is equal to ML decoding of the respective code and the performance of \codename{1} is close to the ML-LB. For all codes, ERD shows a worse performance than ISD. The best results of ERD are achieved by \codename{3}, which is also shown in Fig.~\ref{fig_Hard_63_22}. For small channel error probabilities, the performance is close to ISD of \codename{2}, whereas for large $p$, ISD of all codes shows a better performance.
	The BMD decoding is more than a factor of $100$ worse than code $1$.
	
	For the performance of redundancy set decoding, we simulate two codes. 
	There exists a BCH code of dimension $24$ and minimum distance $15$, which was already used above.
	For this code the cyclotomic coset
	$21$  is exchanged by the coset  $15$ in code $4$ of Example~\ref{ex:rmcode}.
	This increases the coderate at the same minimum distance.
	Fig.~\ref{fig_CohHard_63_24} shows the decoding performance of redundancy set decoding for this code in dependence on the number of cyclic shifts used by the decoder. 
	The number $\mu$ is chosen to be $17$ and no flip-patterns are used.
	Using four shifts of the received vector, the decoding performance is very close to the ML bound. 
	For comparison, the performance of information set decoding of code $1$ from Example~\ref{ex:rmcode} is included.
	As expected, the ML decoding of a code with fewer codewords, i.e. smaller coderate, should have a better performance, and  Fig.~\ref{fig_CohHard_63_24} shows this effect. 
	In general, the redundancy set decoding has a worse performance
	than ISD. 
	
	We use redundancy set decoding for a BCH code of length $n = 127$ and dimension $k=43$.
	The chosen cyclotomic cosets are $K_i, i \in \{1,3,5,7,9,11,13,15,19,27,29,43 \}$, which provides a designed minimum distance of $d = 21$ and a true minimum distance of $\delta = 31$. 
	In this case, there are $63$ cyclically different minimum weight codewords of weight $12$.
	As Fig.~\ref{fig_CohHard_127_43} shows, without additional shifts there is almost no performance loss due to selecting $\mu = 20$ of the most reliable positions from the redundancy part and $k-\mu$ from the systematic for $n=127$.
	This can be explained with the weak law of large numbers.
	Using up to eight shifts and no flip-patterns, the performance of the redundancy set decoder is better than decoding all error patterns up to half the true minimum distance, which is $\lfloor(\delta-1)/2\rfloor=15$.
    It has to be emphasized that algebraic decoding of the used code can only correct up to $10$ errors as it is limited by the designed minimum distance.
	As the performance of the redundancy set decoding is not determined by the designed minimum distance directly, it can improve upon algebraic decoding by a factor of up to $10^5$ at $p=10^{-2}$.
	At the cost of an increased complexity, further performance improvements are possible using ISD with all flip-patterns of weight $\leq j$, to which we refer as ISD($j$).
	ISD($5$) requires testing up to $10^6$ patterns, but approaches the ML performance for channels with a small error probability.

	The redundancy set decoding for the rate one half code was simulated for different $\mu$ and different shifts.
	However, the performance could not reach those of information set decoding
	as in the case of lower rate codes. If no or only a few flip-patterns
	are used, the performance can compete. Nevertheless, there remain many open problems for further research,
	especially the choices of the shifts in combination with $\mu$. 
	There might be advantages in case of longer codes that are not studied yet.

	\begin{figure}[h]
		\centering
		\subfloat[Redundancy set decoding of BCH$(63,24,15)$ \label{fig_CohHard_63_24}]
		{\includegraphics[]{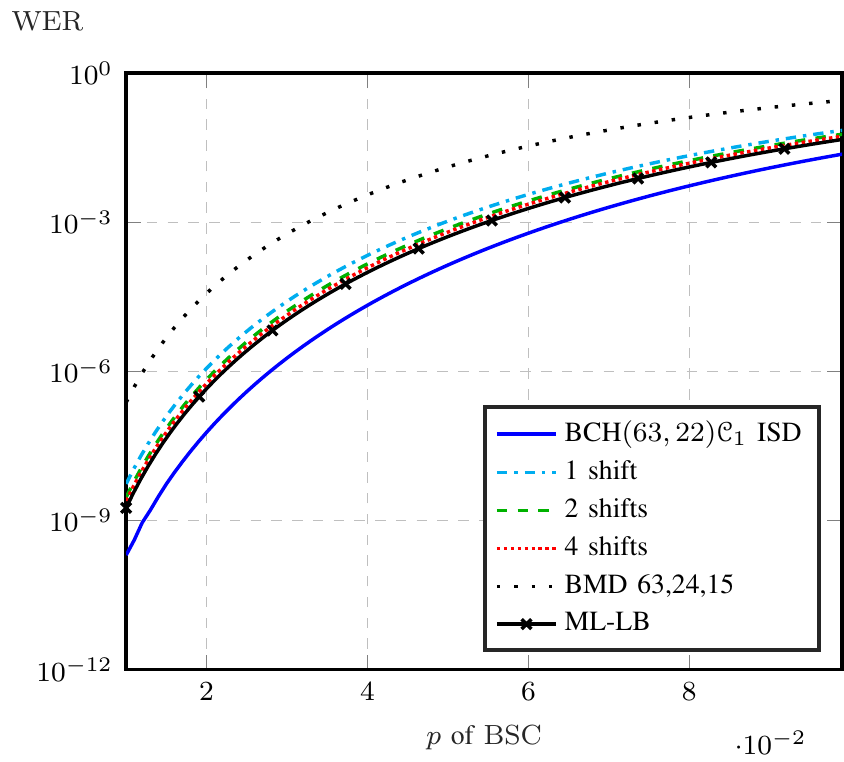}}
		\hfil
		\subfloat[Redundancy set decoding of BCH$(127,43,21)$ \label{fig_CohHard_127_43}]
  		{\includegraphics[]{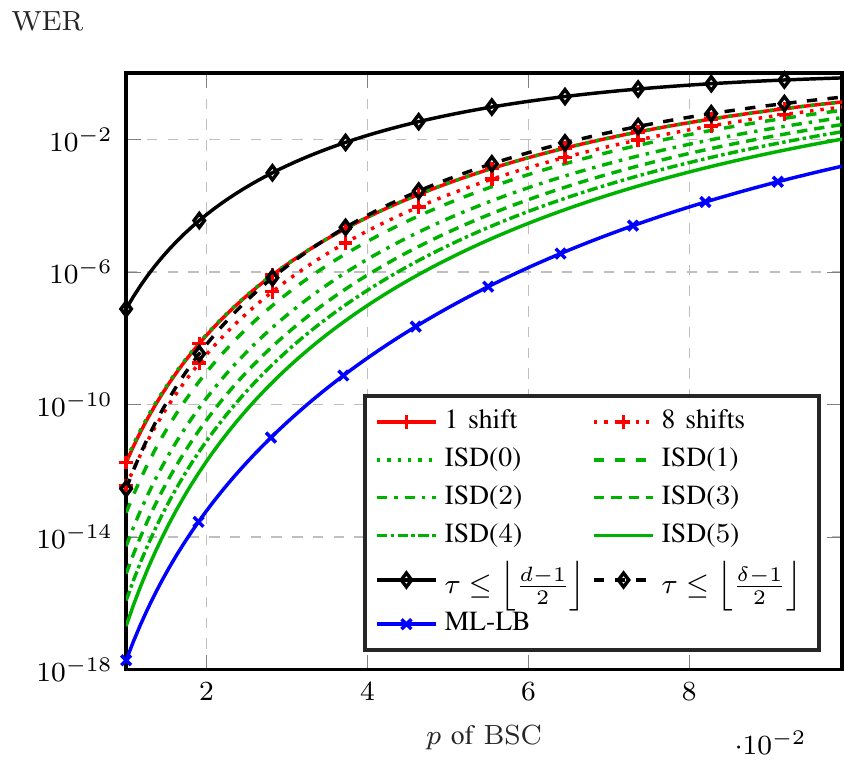}}
		
		\caption{Redundancy set decoding of codes with different codelengths and rate $R \approx 0.35$}
	\end{figure}
	
	\begin{figure}
		\centering
		
		\subfloat[Information set decoding of BCH$(127,64)$ \label{fig_Hard_127_64}]
 		{\includegraphics[]{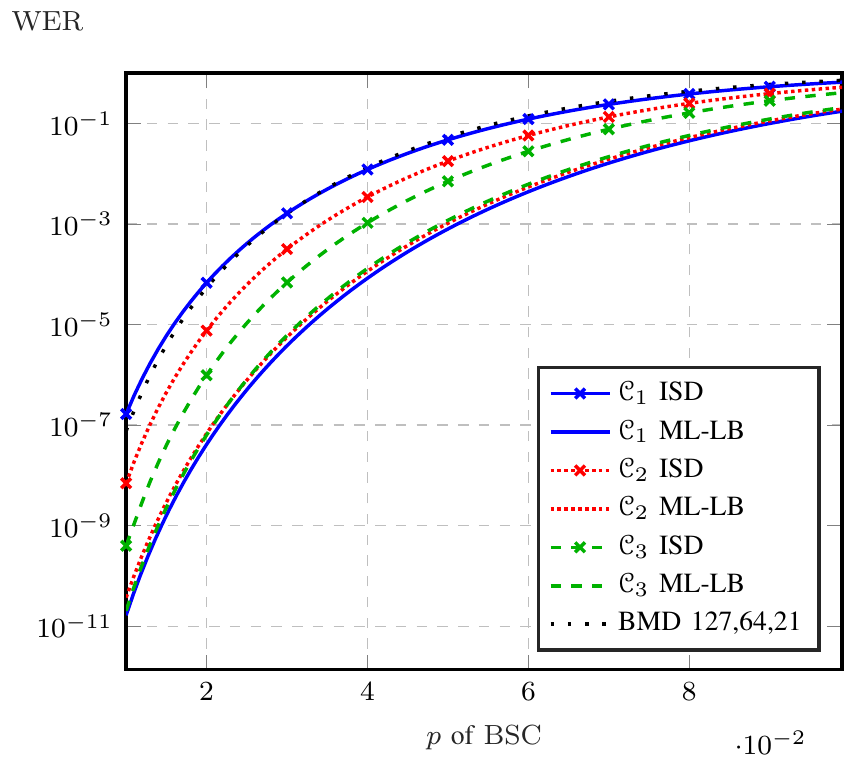}}
		\hfil
		\subfloat[Comparison of BCH and RM \label{fig_Hard_ISD_127_64}]
 		{\includegraphics[]{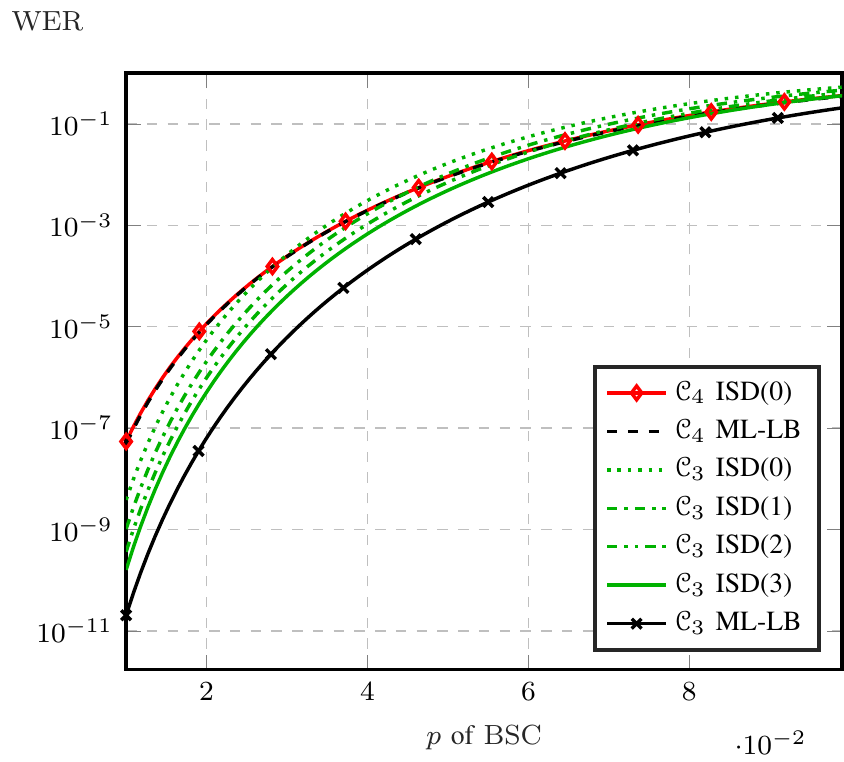}}
		
		\caption{Information set decoding of BCH$(127,64)$}
	\end{figure}
	Next, we look at the performances of the four BCH codes $(127,64)$ from Example~\ref{bch127_64} 
	which have $28$, $119$, $1\,590$, and $651$  decoding polynomials, respectively.
	First, we compare the three BCH codes which are not equivalent to an RM code 
	and then we compare the RM equivalent code \codename{4} to \codename{3}. For ISD, all flip-patterns of weight $\leq 2$ are used, which corresponds to $2\,081$ patterns.	While \codename{1} shows the same performance as BMD decoding for ISD,  codes
	\codename{2} and \codename{3} show better performances, 
	see Fig. ~\ref{fig_Hard_127_64}. 
	The performances are quite different and seem to be related to the number of decoding polynomials. 
	For these codes, the performances are not close to the ML bounds and as expected
	the worse a decoder performs the larger is the distance to the ML bound. Since we show only the lower bound, ML decoding might be closer to the performance of ISD.
	
	The performance of the punctured RM code \codename{4} is compared with the performance 
	of the BCH code \codename{3} in Fig.~\ref{fig_Hard_ISD_127_64}.
	Both codes have the same dimension, but
	the RM equivalent BCH code has minimum distance $d=15$ whereas the
	BCH code has $d=21$. 
	In case of a large channel error probability,
	there is hardly any difference in the performance of these codes.
	However, at small error probabilities, the larger minimum distance of the BCH code becomes obvious
	by a performance improvement up to a factor of $100$.
	A surprising result is that the performance of the information set 
	decoding of the RM equivalent BCH code is independent of the number of flip-patterns used. The performance of ISD(0) is equal to the ML-LB and thus cannot be improved by using additional flip-patterns.
	Here, ISD($j$) means all possible subsets of the $k$ 
	systematic positions with cardinality $\leq j$ are flipped and the list size
	is $\sum_{i=0}^j {k \choose i}$. For \codename{3}, the performance increases if more flip-patterns are used.
	
	\begin{figure}[h]
	
		\centering
		\includegraphics[]{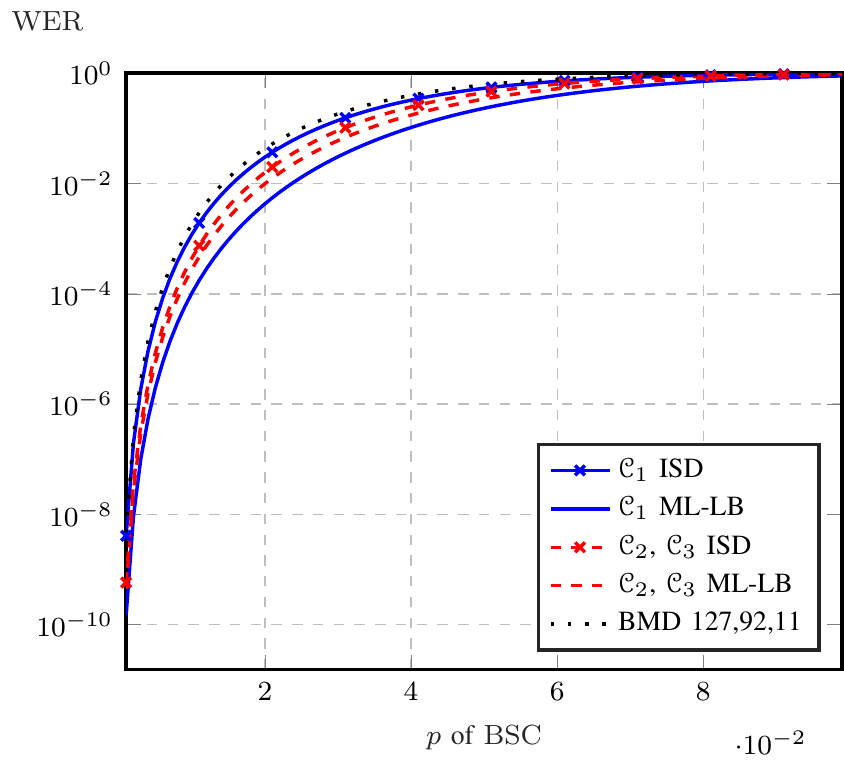}
		\caption{Hard decision decoding of BCH$(127,92)$}
		\label{fig_Hard_127_92}
	\end{figure}
	
	The last example is a BCH code of higher rate, namely  $(127,92,11)$.
	\codename{1} uses the cyclotomic cosets $\{1,3,5,31,63\}$ and has $21$ decoding polynomials,
	\codename{2} uses $\{1,3,5,7,9\}$ and has $155$ polynomials, and
	\codename{3} uses $\{1,3,7,23,55\}$ and has $429$ polynomials. The results for these codes are displayed in Fig.~\ref{fig_Hard_127_92}. For this coderate, $4\,279$ flip-patterns are used for information set decoding. \codename{2} and \codename{3} show identical performance, which is close to the ML bound. In contrast, \codename{1}, which uses only $21$ decoding polynomials, performs worse and is only slightly better than BMD decoding.

	\section{Soft Decision Decoding}\label{sec:softdec}
	For the additive white Gaussian noise (AWGN) channel we use the usual mapping of the binary code symbols  $c_i=0 \leftrightarrow x_i = 1$ 
	and $c_i=1 \leftrightarrow x_i = -1$ which corresponds to binary phase shift keying (BPSK).
	We receive $y_i = x_i + n_i$ where $n_i$ denotes the Gaussian noise.
	The signal to noise ratio for information bits in dB is called $E_b/N_0$. 
	Then the signal to noise ratio  for a codebit is $E_s/N_0 = E_b/N_0 - 10 \cdot \log_{10}(1/R)$
	where $R$ is the coderate of the used code.
	The variance of the noise is then $\sigma^2=1/(2 \cdot 10^{E_s/(10 N_0)})$.
	
	First, we will describe the intrinsic reliability information from the channel. 
	Then, we use the dual codewords to calculate extrinsic reliability information and show that this improves the reliability information.
	A statistical analysis of the novel reliability information can be used to predetermine flip-patterns which are used for list information set decoding. 
    An example that this method allows to predict the  word error rate performance of the decoding is given.
	Then, the used information set decoder is described.
	We give no results of the soft decision error reduction decoder since the performance
	is worse than with the information set decoding.
	This can be intuitively explained with the same observation that we made for the BSC in Example \ref{ex:err_distr}: The probability that most of the reliable positions are correct is higher than the probability that the least reliable positions are erroneous.
	The simulation results show that the reliability from the Gaussian channel can be improved by the inherent reliability from the minimal weight dual codewords and, thus, the decoding performance.
	We illustrate that, when taking only the check polynomial the decoding performance can not be improved compared to only using the intrinsic reliability from the channel.
	From all checks, only some may be selected according to a criterion which leads to a reduction in decoding complexity
	which was also observed in \cite{haeger}.
	We also show that increasing the number of flip-patterns improves the decoding performance until ML decoding is reached.
	Finally, we compare our decoder to adaptive belief propagation \cite{ABP}, bias-based multibasis information set decoding \cite{bias_ISD}, and to decoding of polar codes \cite{Liva} and give runtime measures to illustrate and compare the decoding complexities.
	We did not use stopping criteria from the literature in order to reduce the average decoding time, however, this could be included straightforward into the algorithms.

	\subsection{Aspects of Reliability Information}
	While the BSC offers no information about the reliability of the received symbols,
	the AWGN channel with its real-valued output alphabet
	allows the calculation of  reliability information and we use the common measure 
	(\cite{bias_ISD, ABP})
	\begin{equation}
	L_j = \tanh\left(\frac{y_j}{\sigma^2}\right)\;.
	\end{equation}
	Clearly, if $|L_i|>|L_j|$, then position $i$ is more reliable than position $j$.
	The reliabilities can be used for information set decoding. According to Algorithm~\ref{Alg-ISD}, among the most reliable 
	positions, $k$ are chosen which are an information set. We will denote this by ISD-Chan. 
	However, the minimal weight codewords of the dual code define a set of checks which can be used to calculate extrinsic reliability information.
	Let $j_0, j_1, \ldots, j_{\delta^\perp -1}$ be the non-zero positions of the dual codeword
	(as vector, not as polynomial).
	The product of the $L_j$ is denoted by $\Delta = L_{j_0} \cdot  L_{j_1} \cdot \ldots  \cdot L_{j_{\delta^\perp -1}}  $.
	Then, the extrinsic reliability for these check positions is 
	\begin{equation}
		\Phi_i = 2 \atanh\left(\frac{\Delta}{L_i}\right), \ i = j_0, j_1, \ldots, j_{\delta^\perp -1} \;.
	\end{equation}
	All minimal weight dual codewords are used as checks and the extrinsic information for each position is added.
	Now, the reliability information is the combination of the intrinsic information (channel reliability $L_j$)
	and the extrinsic information from all checks $L_j + \alpha \Phi_j$, where $\alpha$ is 
	an attenuation factor which is determined by simulations.
	Again, the most reliable positions can be chosen to find an information set which will be denoted by ISD-Dual.
	Does the extrinsic information improve the quality of the reliability information?
    In order to visualize this, we simulated the number of errors in the best positions for ISD-Chan and ISD-Dual.
    The simulation is for the BCH$(127,64)$ code \codename{3} at \SI{2}{\dB}. 
    This code has $1\,590$ minimal weight dual code polynomials which corresponds to $201\,930$ checks.
    Fig.~\ref{fig_Matrix_FlipPattern} shows the curves that $\tau$ errors occur in the $k$ best positions. 
    Each point $a_{\tau,\ell}$ of a curve is the relative frequency that $\tau$ errors have occured and $\ell$ determines the most reliable of the $\tau$ errors.
    So $\tau-1$ errors are in the more unreliable positions $\ell+1$ to $k-1$.
    Therefore, the probability that $\tau$ errors occur in the positions $\ell$ to $k-1$ is
 $a_{\tau,\ell}+a_{\tau,\ell+1}+ \ldots+a_{\tau,k-1}$.
    It can be seen that the number of errors in the best positions is considerably reduced for ISD-Dual compared to ISD-Chan. 
    For example, for $\ell=15$ and $\tau=3$ the relative frequency that $2$ errors are among  
the positions $16$ to $63$ for ISD-Chan is the same as  $\tau=1$  error at position $\ell=15$ 
when using ISD-Dual. Similarly, for $\ell=50$ and $\tau=2$ 
the relative frequency that $1$  of the two errors is in positions $51$ to $63$ for ISD-Chan is larger than
there is $\tau=1$ error in position $50$ for ISD-Dual. Further, for ISD-Dual the probability 
of  $\tau=1$  error in the  positions $0$ to $6$ is the sum of the values 
of the relative frequencies from $0$ to $6$ and it is smaller than the probability of  $\tau=3$  errors 
in the positions $42$ to $63$.
This fact will be used for the selection of the flip-patterns for list information set decoding in the following.	
	
	\begin{figure}
		\centering
                \includegraphics[]{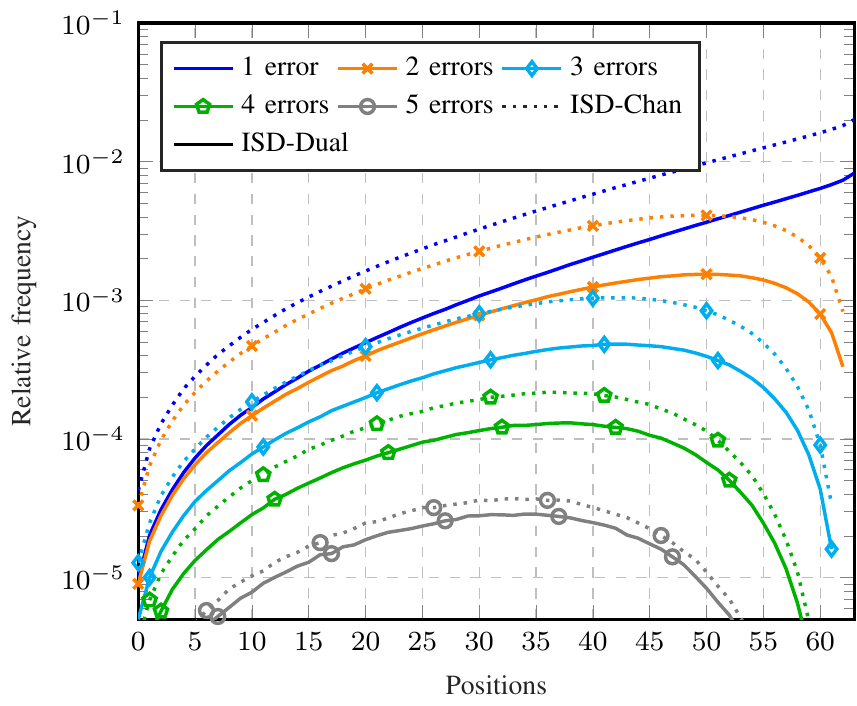}
		\caption{Relative frequency of errors in the $k$ positions that are contained in the information set for code BCH$(127,64)$ \codename{3}. 0 corresponds to the most reliable position and 63 is the position with the smallest reliability. Simulated at \SI{2}{\dB} using $T=100$ and $\alpha = 0.07$.}
		\label{fig_Matrix_FlipPattern}
	\end{figure}

For list decoding, flip-patterns are added to the positions of the information set. There exist several strategies for choosing flip-patterns given a fixed list size.
	In \cite{flippattern} several strategies are compared and one is based
	on the so called elliptic weight.
We use a similar approach as in \cite{wu2008soft} based on Fig.~\ref{fig_Matrix_FlipPattern}.  We construct a matrix 
	$A$ which counts the relative frequency of occurrence of  $\tau$  errors at certain positions. 
	As defined before, each element $a_{\tau,\ell}$ of $A$ gives the relative frequency of $\tau$ errors in the $\ell$ 
	positions of the information set where $\ell$ is the error with the largest reliability value. 
	With this matrix $A$, the expected WER of a simulation can be predicted by 
	the sum over all elements $a_{\tau,\ell}$, which are not within the flip-patterns
	\begin{equation}\label{expwer}
	\text{WER}_{est} = \sum_{\tau}\sum_{\ell=0}^{k-1} a_{\tau,\ell}\;.
	\end{equation}
	This estimation can, on the other hand, also be used to calculate how many flip-patterns have to be used in order to reach a certain performance.
	For example in case of two errors the relative frequency that the more reliable one
	of these two is at position $30$ in the information set is according Fig.~\ref{fig_Matrix_FlipPattern} about $10^{-3}$.
	It can be observed that the relative frequency of $3$ errors in positions $50$ to $64$,
	which are the unreliable positions among the $k$ most reliable ones,
	is higher than for $2$ errors with the larger reliability in the positions $0$ to $14$. 
	Equation~(\ref{expwer}) can also be used to minimize the expected WER under the condition that a fixed number of flip-patterns is used.
	The following example compares the WER simulation with the estimation using the preselected flip-patterns.
	\begin{example}[Estimated WER]
            The matrix is simulated for the code BCH$(127,64,21)$  \codename{3} at \SI{2}{\dB}.
            If a maximum number of $100$ flip-patterns is to be used, the matrix indicates that it is better to use all flip-patterns of weight one on the $55$ least reliable positions and all flip-patterns of weight two on the ten least reliable positions.
            The matrix indicates that a WER of $\text{WER}_{est} = 0.060$ is to be expected.
            The true WER obtained from simulations with the chosen flip-patterns is $0.064$.
	\end{example}

	\subsection{Decoding Algorithm}
	This section describes the soft decision list information set decoding algorithm.
	We assume that the vector $y = (y_0, y_1, \dots, y_{n-1})$ has been received after transmission over the AWGN channel. The used reliability is $L_j = \tanh\left(\frac{y_j}{\sigma^2}\right)$.
	The reliabilities of a check are
	\begin{equation}
	L_{\sup(b)} = \{L_{b_0}, L_{b_1}, \dots, L_{b_{\delta^\perp-1}}\}\;.
	\end{equation}
    We introduce a modification which does not use checks with too many unreliable positions.
    For this, we find the  $T$ most reliable positions. 
    We use only checks that have $\geq \delta^\perp -1$ positions within these $T$ most reliable positions.
    In other words, only zero or one position of the check is from the $n-T$ most unreliable positions.
    The following algorithm is used for the simulations in the next section.  	

\vspace{0.5cm}
\begin{algorithm}[h]
		\DontPrintSemicolon
		\SetKwInput{Input}{input}\SetKwInput{Output}{output}
		\Input{$y$, minimum weight checks $\mathcal{W}$, $T\in \mathbf{N}$, flip-patterns $\mathcal{F}$, damping coefficient $\alpha$}
		\Output{$\hat{c}$}
		
		\tcp{intrinsic information}
		$L_j =  \tanh\mleft({\frac{y_j}{\sigma^2}}\mright)\quad\forall j $\;
		\tcp{calculate $\Phi$}
		$\Phi_j = 0 \quad \forall j$\;
		find the $T$ best positions \; 
		\For{$\forall\mathcal{H}\in\mathcal{W}$}{
			\If{ $ \geq \delta^\perp-1$ positions of $\mathcal{H}$ are in the $T$ best}{
				$\Delta = \prod_{h\in\mathcal{H}} L_h$\;
				\For{$h\in\mathcal{H}$}{
					\tcp{extrinsic information}
					$\Phi_h \pluseq 2\cdot\atanh\mleft(  \Delta / L_h \mright)$ 
				}	
			}	
		}
		\tcp{update intrinsic information}
		$
		L_j \pluseq \alpha\cdot \Phi_j \quad\forall j 
		$
		
		\tcp{perform information set decoding}
		$G_{sorted} = $ sort columns of $G$ according to reliability of $L_j$\;
		$I = $ pivot columns of $G_{sorted}$\;
		$r = \mathrm{hard}(L)$\;
		
		\tcp{apply flip-patterns} 
		\For{$\forall e \in \mathcal{F}$}{
			$r_I = r|_I + e$\;
			$\hat{c} = r_I \cdot G_I^{-1} \cdot G$\;
		}
		\Return $\hat{c} = \mathrm{argmin}\left(\sum_j \left(r_j-(-1)^{\hat{c}_j}\right)^2\right)$
		\vspace{4pt}
		\caption{Decoding algorithm ISD-Dual with $\alpha$ and $T$. If $T$ is omitted, all checks are used.\label{LISD}}
	\end{algorithm}

	\subsection{Decoding Performance}

    \begin{figure}
		\centering
				\includegraphics[]{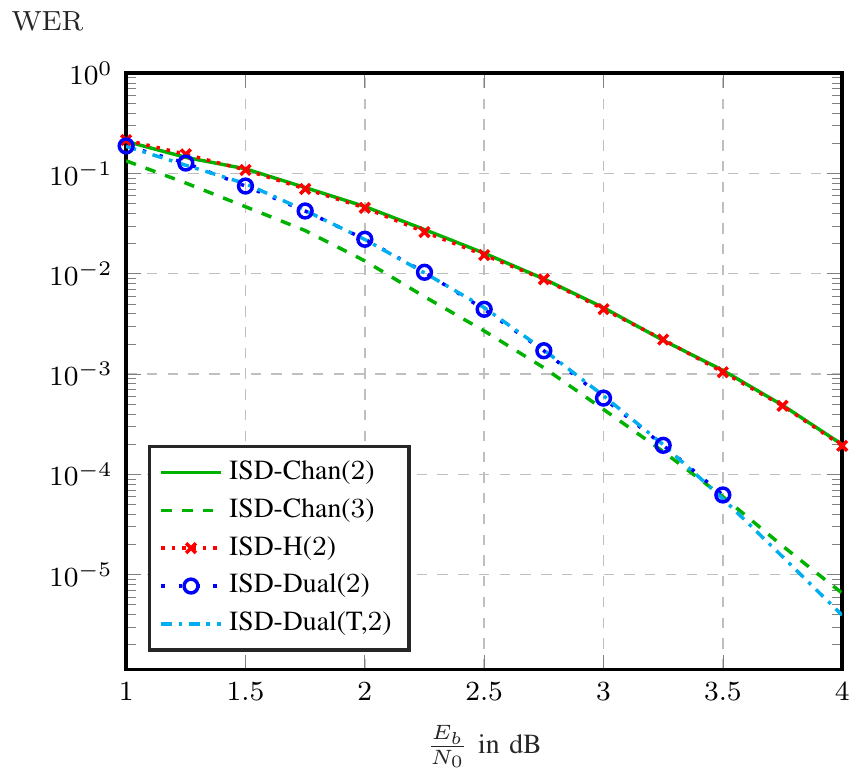}
		\caption{Comparison of different approaches including extrinic information into the decoding of the BCH$(127,64,21)$ code \codename{3}.}
		\label{fig:BCH1}
	\end{figure}
	
	In Fig.~\ref{fig:BCH1}, we show the simulation results of soft decision list information set decoding of the BCH$(127,64,21)$ code \codename{3}. 
	The WER curves for ISD-Chan using $2\,080$ flip-patterns of weight $\leq 2$ and for $43\,744$ of weight $\leq 3$ are depicted.
	The curve ISD-H uses the check polynomial which corresponds to $n$ checks
	to calculate extrinsic reliability information with $\alpha = 0.2$ and all flip-patterns of weight $\leq 2$. 
	It can be observed that this does not improve the decoding performance compared to ISD-Chan.
	Using all dual minimal weight checks ISD-Dual with  $\alpha = 0.07$ and all flip-patterns of weight $\leq 2$ considerably improves the decoding performance. 
	In fact, it crosses the  ISD-Chan with all flip-patterns of weight $\leq 3$ at about \SI{3.5}{\dB}.
	Another observation is, when using the threshold $T=100$ the number of checks is 
	reduced from $201\,930$ to $5\,089$ in average without any loss in performance.  
	This reduces the decoding complexity considerably, as shown later.

	In Fig.~\ref{fig:BCHRM1}, the BCH$(127,64,21)$ code \codename{3} and the punctured RM code 
	$\codename{4}$
	are compared with soft decision list information set decoding. 
	For all simulations flip-patterns of weight $\leq 2$ are used. 
	The first curve shows the ISD-Chan for the BCH code.
	For all remaining curves $T=100$ is used.
	The RM code (ISD-Dual-RM)  with $\alpha = 0.1$ shows a performance close to the ML bound (ML-LB-RM) and performs slightly better than the BCH code (ISD-Dual-BCH) for \SI{<2.7}{\dB}.
	However, since the BCH has a larger minimum distance, the decoding performance
	gets better than the ML decoding of the RM code for better channel conditions.
    Further, since the (ML-LB-BCH) shows a much better performance, the decoding of the BCH code can be improved by using more flip-patterns, which is shown later in  Fig.~\ref{fig:BCH2}. 
	The ML bounds are different for the RM and the BCH code which is a property of the bound that it estimates the ML performance of the particular code used.
    An interesting observation is that the decoding of the RM code can not be improved by using
	more flip-patterns since it is already very close to the ML bound.
	Note that the slope of ISD-Dual-BCH is steeper than the one of ISD-Chan-BCH.

    \begin{figure}
		\centering
		\subfloat[Comparison of the BCH$(127,64,21)$ code \codename{3} and the punctured RM code \codename{4}.\label{fig:BCHRM1}]
        {\includegraphics[]{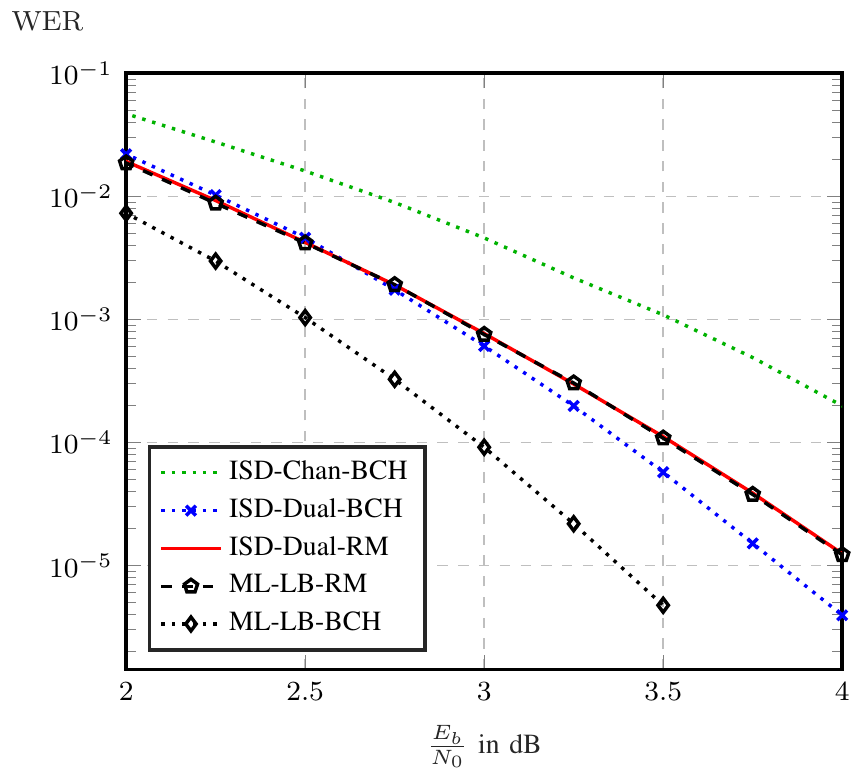}}
		\hfil
		\subfloat[Decoding the BCH$(127,64,21)$ code \codename{3} with different number of  flip-patterns. \label{fig:BCH2}]
 		{\includegraphics[]{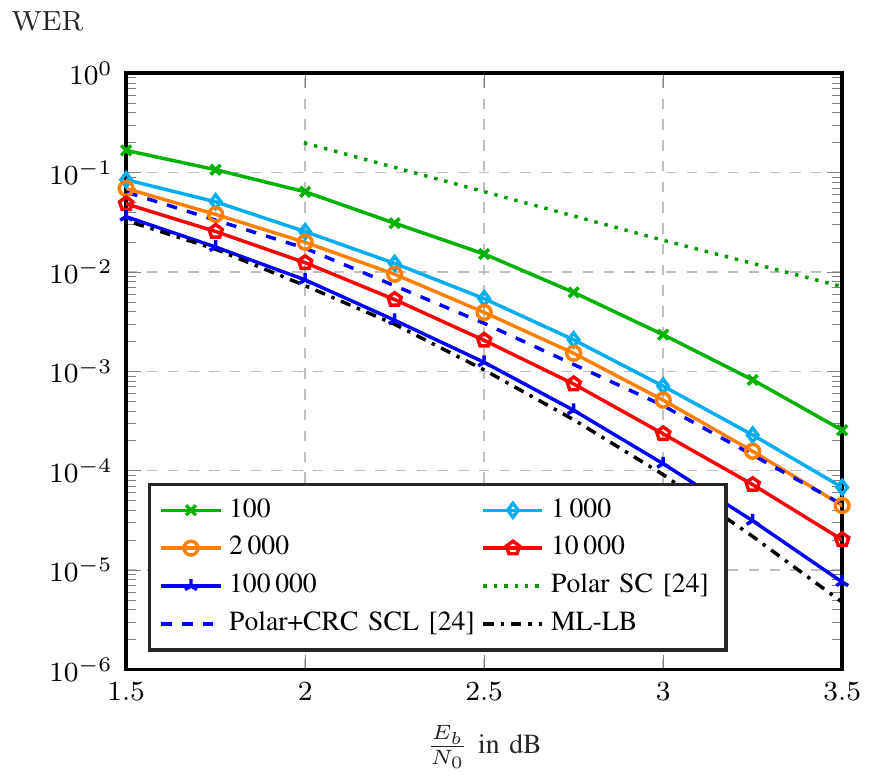}}
		\caption{Soft decision decoding with the proposed decoding algorithm.}
	\end{figure}

	Now, we show that increasing the number of flip-patterns improves the decoding performance.
	Different list sizes are compared in Fig.~\ref{fig:BCH2}.
	Again, the BCH$(127,64,21)$ code \codename{3} is simulated.
    We use ISD-Dual with $T=100$ and $\alpha=0.07$.
    The flip-patterns are selected according to (\ref{expwer}) in order to minimize the expected WER.
	With $100$ flip-patterns the performance is already better than the polar code of length $128$ and $k=64$ from \cite{Liva} with successive cancellation (SC) decoding.  
	With  $10\,000$  flip-patterns the performance is strictly better than list decoding of the crc-aided polar code from \cite{Liva}.
	Further, the performance curve of the crc-aided polar code is flatter, because it has a smaller minimum distance.
	As a consequence, already $2\,000$  flip-patterns suffice in order to achieve the same performance with the BCH code at \SI{3.5}{\dB}.
	Using  $100\,000$ flip-patterns the performance is close to ML decoding.
	The decoded code is the best known linear code of this length and rate \cite{Grassl:codetables}.
	Therefore, it is not expected that any other code can achieve a better performance and we omit a comparison with further codes.

	\begin{figure}[h]
		\centering
				{\includegraphics[]{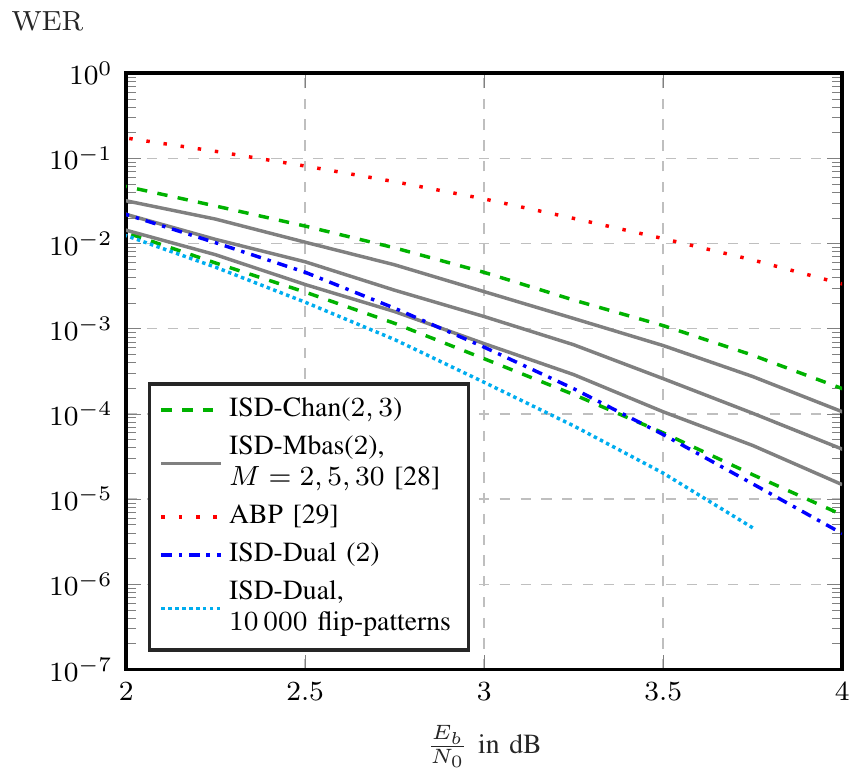}}
		\caption{Comparison of the proposed algorithm with adaptive belief propagation \cite{ABP} and bias-based multibasis information set decoding \cite{bias_ISD}.}\label{fig:soft_comparison}
	\end{figure}
    In	Fig.~\ref{fig:soft_comparison}, we compare our ISD-Dual decoding with ISD-Mbas (bias-based multibasis information set decoding \cite{bias_ISD}) and with ABP (adaptive belief propagation  \cite{ABP}). 
    It can be seen that ABP with $M=20$ iterations and damping coefficient $\alpha=0.1$ has the worst performance. 
    In addition, we give runtime measures for these algorithms in Table~\ref{table:runtimes}. 
    Three curves are plotted for ISD-Mbas with $a=0.1$ for a different number of bases  $M=2,5,30$  and the related runtime measures are given in Table~\ref{table:runtimes}. 
    The best performance shows the ISD-Dual with $T=100$ and $\alpha=0.07$ with  $10\,000$ flip-patterns at a factor of $5$ smaller runtime as the $30$ basis ISD-Mbas. 
    The performance is also better than ISD-Chan with all flip-patterns of weight $\leq 3$ which has almost doubled runtime. Note that the slope of ISD-Dual is steeper than both, ISD-Chan and ISD-Mbas.

    \begin{table*}
        \centering
        \caption{Measured worst-case runtimes of the proposed decoding algorithm, adaptive belief propagation \cite{ABP} and bias-based multibasis decoding \cite{bias_ISD}.}\label{table:runtimes} 
        \begin{tabular}{cccccccccc} 
        \toprule
        \multicolumn{2}{c}{ISD-Chan} & \multicolumn{4}{c}{ISD-Dual($2$) with $\alpha = 0.07$} & \multicolumn{3}{c}{ISD-Mbas($2$), $a=0.1$}& ABP $M=20$, $\alpha=0.1$ \\ 
        \cmidrule(lr){1-2}\cmidrule(lr){3-6}\cmidrule(lr){7-9}
        &&\multicolumn{2}{c}{$w=2$} & \multicolumn{2}{c}{$10\,000$ flip patterns}\\
        \cmidrule(lr){3-4}\cmidrule(lr){5-6}
        $w = 2$ & $w=3$ & $T=100$& full & $T=100$& full & $M=2$ & $M=5$& $M=30$\\
        \midrule 
        \SI{1.2}{\milli\second}&\SI{13.3}{\milli\second}&\SI{4.8}{\milli\second}&\SI{28.3}{\milli\second}&\SI{7.0}{\milli\second}& \SI{30.5}{\milli\second} &  \SI{2.5}{\milli\second}&\SI{6.3}{\milli\second}&  \SI{37.1}{\milli\second} & \SI{11.5}{\milli\second} \\
        \bottomrule
        \end{tabular}
    \end{table*}

	\section{Conclusions}
	We have shown that erroneous codewords can be checked by minimal weight dual codewords and these checks 
	provide reliability information for the received symbols.
	We described three decoders which use this reliability and information set decoding was superior.
	In many examples the decoding has the same performance as hard decision ML decoding
	which we have shown using an ML bound.
	We recalled that any RM code can be described as a particular BCH code extended by a parity bit.
	We presented examples where the BCH codes show better performance than
	the punctured RM codes. 
    We described BCH codes with different choices of the cyclotomic cosets and compared their decoding perormance.
    For soft decision list information set decoding, we have derived that the use of the extrinsic reliability from dual codewords in addition to the channel reliability can improve the decoding performance considerably.
	We compared the performance to known methods from literature and
	for the BCH code of length $127$ and dimension $64$ our decoder can reach close to ML decoding.
	
	\begin{appendix}
		Reed--Muller codes can be defined recursively using the Plotkin construction \cite{Plotkin}, however, the same set of codewords can also be obtained from a definition based on Boolean functions. 
		We refer to a function $\mathrm{f}(v_0,\ldots,v_{m-1})$ as a Boolean function if it takes on the values $0$ and $1$ for any input tuple in $\mathbb{F}_{2}^m$.
		A Reed--Muller code of length $2^m$ is then the set of all Boolean functions of degree up to $r$ evaluated at all $2^m$ input $m$-tuples:
		\begin{align*}	
		\mathcal{R}(r,m) = &\left\{\left(\mathrm{f}(i_0,\ldots,i_{m-1})\right)_{i=0}^{2^m-1} \mid \mathrm{f}(i_0,\ldots,i_{m-1}) \vphantom{\sum_{\boldsymbol s} C_{\boldsymbol s}}\right.\\
		&\phantom{\{\}}=\left. \sum_{\boldsymbol s} C_{\boldsymbol s} (i_0)^{s_0}\cdot\ldots\cdot(i_{m-1})^{s_{m-1}}\right\}\;,
		\end{align*}
		where $i = i_0 + i_1\cdot2+\ldots+i_{m-1}2^{m-1}$, the sum runs over all $\boldsymbol s=(s_0,\ldots,s_{m-1}) \in \{0,1\}^m$ such that $0\leq\sum_{t=0}^{m-1} s_t \leq r $ and $C_{\boldsymbol s} \in \mathbb F_{2}$.
		
		In the definition of the Reed--Muller code, the Boolean functions were evaluated at $(i_0,\ldots,i_{m-1})$, where each $m$-tuple corresponded to the binary expansion of the respective code coordinate.
		However, now, we change the evaluation order and evaluate a Boolean function at $(a^{(i)}_{0},\ldots,a^{(i)}_{m-1})$, the coefficients of $\alpha^i \in \mathbb{F}_{2^m}^*$ when expressed through the polynomial basis $\{\alpha^0,\ldots,\alpha^{m-1}\}$, in order to obtain the value of the codeword at position $i$.
		Let $\alpha^{i}=a^{(i)}_{0}+a^{(i)}_{1}\alpha+\ldots+a^{(i)}_{m-1}\alpha^{m-1}$, the code 
		\begin{align*}
		\mathcal{R}(r,m)^* = &\left\{\left(\mathrm{f}(a^{(i)}_{0},\ldots,a^{(i)}_{m-1})\right)_{i=0}^{2^m-2} \mid  \mathrm{f}(a^{(i)}_{0},\ldots,a^{(i)}_{m-1}) \vphantom{\sum_{\boldsymbol s} C_{\boldsymbol s}}\right.\\
		&\phantom{\{\}}\left.= \sum_{\boldsymbol s} C_{\boldsymbol s} (a^{(i)}_{0})^{s_0}\cdot\ldots\cdot(a^{(i)}_{m-1})^{s_{m-1}}\right\}
		\end{align*}
		is then a punctured and permuted RM code.
		First, we show that $\mathcal{R}(r,m)^*$ is cyclic.
		This is equivalent to showing that for all $\mathrm{f}(v_0,\ldots,v_{m-1})$, there exists a $\tilde{\mathrm{f}}(v_0,\ldots,v_{m-1})$, which satisfies
		\begin{align}
		\mathrm{f}(a^{(i)}_{0},\ldots,a^{(i)}_{m-1}) = \tilde{\mathrm{f}}(a^{(i+1)}_{0},\ldots,a^{(i+1)}_{m-1})		
		\end{align}
		$\forall i \in \{0,\ldots,2^m-2\}$.
		In $\mathbb{F}_{2^m}$, the multiplication of $\alpha^{i}$ with $\alpha$ is carried out modulo a primitive polynomial $\mathrm{p}(x)= x^m+\sum_{\ell=1}^{m-1} \mathrm{p}_{\ell}\cdot x^{\ell}+1$ with coefficients in $\mathbb{F}_{2}$. Therefore, we arrive at
		\begin{align}
		\mathrm{f}(a^{(i)}_{0},\ldots,a^{(i)}_{m-1}) &=
		\mathrm{\tilde{f}}\left(a^{(i)}_{m-1},a^{(i)}_{0}+\mathrm{p}_{1}a^{(i)}_{m-1},a^{(i)}_{1}+\mathrm{p}_{2}a^{(i)}_{m-1},\right.\nonumber\\
		&\phantom{=,}\left.\ldots,a^{(i)}_{m-2}+\mathrm{p}_{m-1}a^{(i)}_{m-1}\right).
		\end{align}
		By applying substitution, one sees that a $\mathrm{\tilde{f}}(v_0,\ldots,v_{m-1})$ which satisfies this for all $i$ is given by
		\begin{equation}
		\mathrm{\tilde{f}}(v_0,\ldots,v_{m-1}) = 
		\mathrm{f}(v_1+\mathrm{p}_1v_0,v_1+\mathrm{p}_2v_0,\ldots,v_1+\mathrm{p}_{m-1}v_0,v_0)
		\end{equation}
		and is of the same degree as $\mathrm{f}(v_0,\ldots,v_{m-1})$.
		Hence, $\mathcal{R}(r,m)^*$ is indeed cyclic.
		For $\mathcal{R}(r,m)$ the order of the evaluation tuples is given by the natural representation of the respective position as binary number.
		For the cyclic code $\mathcal{R}(r,m)^*$ the order is the vector representation of $\alpha^{i}$ using the polynomial basis.
		This proves the permutation given in \ref{subsec:equivalence}: $\pi$ maps $i$ to $j$ such that the binary expansion of $i$ is identical to the coefficients of $\alpha^j$ under the polynomial basis.
		
		Now, we find a generator polynomial of $\mathcal{R}(r,m)^*$.
		We need to find the roots of $\sum_{i=0}^{2^m-2} \mathrm{f}(a^{(i)}_{0},\ldots,a^{(i)}_{m-1})\ \alpha^{h\cdot i}$.
		Let $h = \sum_{t=0}^{m-1}h_t 2^t$ with $1\leq \sum_{t=0}^{m-1}h_t \leq m-r-1$ and $h_t\in\{0,1\}$, then
		\begin{align}
		\alpha^{h\cdot i} 
		&=\left(\sum_{j=0}^{m-1} a^{(i)}_{j}\alpha^{j} \right)^{\sum_{t=0}^{m-1}h_t 2^t}\nonumber\\
		&=\prod_{t=0}^{m-1}\left(\sum_{j=0}^{m-1} a^{(i)}_{j}\alpha^{j} \right)^{h_t2^t}\nonumber\\
		&=\prod_{\{t:\ h_t=1\}} \sum_{j=0}^{m-1} a^{(i)}_{j}\alpha^{j2^t} 
		\end{align}
		Following \cite{Kasami}, we can expand the product as
		\begin{equation}
		\alpha^{h\cdot i} = \sum_{\boldsymbol \ell}\left(a^{(i)}_0\right)^{\ell_0}\cdot\ldots\cdot\left(a^{(i)}_{m-1}\right)^{\ell_{m-1}} \alpha^{\beta(h,\boldsymbol \ell)}\;,
		\end{equation}
		where the sum runs over all $\boldsymbol \ell=(\ell_0,\ldots,\ell_{m-1})$ s.t. $1\leq\sum_{t=0}^{m-1} \ell_t = \sum_{t=0}^{m-1}h_t \leq m-r-1$.
		Utilizing that $\beta(h,\boldsymbol \ell)$ does not depend of $j$, we can reorder the sums and obtain
		\begin{align}
		\sum_{i=0}^{2^m-2} \mathrm{f}(a^{(i)}_{0},\ldots,a^{(i)}_{m-1})\ \alpha^{h\cdot i}
		= \sum_{\boldsymbol \ell} \alpha^{\beta(h,\boldsymbol \ell)} \sum_{\boldsymbol s} C_{\boldsymbol s} \nonumber\\\sum_{j=0}^{2^m-2} (a^{(i)}_{0})^{s_{0}+\ell_{0}}\cdot\ldots\cdot(a^{(i)}_{m-1})^{s_{m-1}+\ell_{m-1}},
		\end{align}
		which is $0$ as the innermost sum is $0$:
		Due to the restrictions on $\boldsymbol \ell$ and $\boldsymbol s$, there is at least one $\gamma$ such that $s_\gamma+\ell_\gamma=0$.
		The sum goes through all $m$-tuples $(a^{(i)}_{0},\ldots,a^{(i)}_{m-1})$ in $\mathbb{F}_{2^m}^*$.
		Therefore, the summands $(a^{(i)}_{0})^{s_{0}+\ell_{0}}\cdot\ldots\cdot(a^{(i)}_{\gamma})^0\cdot\ldots\cdot(a^{(i)}_{m-1})^{s_{m-1}+\ell_{m-1}}$ and $(a^{(i)}_{0})^{s_{0}+\ell_{0}}\cdot\ldots\cdot(a^{(i)}_{\gamma}+1)^0\cdot\ldots\cdot(a^{(i)}_{m-1})^{s_{m-1}+\ell_{m-1}}$ add up to $0$.
		The $m$-tuple which has only a $1$ in position $\gamma$ would correspond to the all-zero tuple which is not in $\mathbb{F}_{2^m}^*$, note however that its summand is already $0$, as there is at least one $\theta$ s.t. $s_\theta+\ell_\theta\neq0$, because $h>0$ (in other words, $h=0$ has to be excluded as there would be a sum of an odd number of ones for $\ell=(0,\ldots,0)$).
		Hence, we have seen that all $\alpha^h$ with $1\leq\sum_{t=0}^{m-1}h_t \leq m-r-1$ are roots of all codepolynomials, respectively the generator polynomial. 
		The number of $h$ which satisfy this restriction is
		\begin{align*}
		\binom{m}{1} +\ldots+ \binom{m}{m-r-1} 
		= \binom{m}{m-1} +\ldots+ \binom{m}{r+1} \\
		= 2^m - 1 -\left[\binom{m}{0} +\ldots+ \binom{m}{r}\right] 
		= 2^m -1 - \dim{\mathcal{R}(r,m)^*}.
		\end{align*}
		Hence, there can be no further roots of the generator polynomial.
	\end{appendix}

	\vspace{.2cm}


	\begin{IEEEbiographynophoto}
		{Martin Bossert} (IEEE M'94–SM'03–F'12) received his Dipl.-Ing. degree in Electrical Engineering from the Technical University of Karlsruhe, Germany in 1981, and the Ph.D. from the Technical University of Darmstadt, Germany in 1987. After a one-year DFG scholarship at Linköping University, Sweden, he joined AEG Mobile Communication, where he was involved in the specification and development of the GSM system. Since 1993 he is a Professor at Ulm University, Germany, presently as senior professor of the Institute of Communications Engineering. He is author of several textbooks and co-author of more than 200 papers. He has been a member of the IEEE Information Theory Society board of governors from 2010 to 2012 and has been appointed as a member of the German National Academy of Sciences in 2013. Among other awards and honours, he received the Vodafone Innovationspreis 2007. His research interests are in reliable and secure data transmission. His main focus is on decoding of algebraic codes with reliability information and coded modulation.
		
		\end{IEEEbiographynophoto}
	
		\begin{IEEEbiographynophoto}
		{Rebekka Schulz} received the B.Sc. and M.Sc. degrees in electrical engineering and information technology from  Ulm University, Germany, in 2017 and 2020, respectively. She is currently pursuing the Ph.D. degree with the Institute of Communications Engineering, Ulm University, Germany. Her research interests include Physical-Layer Security and Signal Processing.
	\end{IEEEbiographynophoto}

	\begin{IEEEbiographynophoto}
	{Sebastian Bitzer} (Student Member, IEEE) received the B.Sc. and
	M.Sc. degrees in electrical engineering from Ulm University, Germany, in 2018 and 2021, respectively.
	He is currently pursuing the Ph.D. degree with the Coding and Cryptography
	Group, Institute of Communications Engineering, Technical University of
	Munich (TUM), under the supervision of Prof. Wachter-Zeh.
	His research interests include coding theory and cryptography.
\end{IEEEbiographynophoto}
	

\begin{thebibliography}{10}
		\providecommand{\url}[1]{#1}
		\csname url@samestyle\endcsname
		\providecommand{\newblock}{\relax}
		\providecommand{\bibinfo}[2]{#2}
		\providecommand{\BIBentrySTDinterwordspacing}{\spaceskip=0pt\relax}
		\providecommand{\BIBentryALTinterwordstretchfactor}{4}
		\providecommand{\BIBentryALTinterwordspacing}{\spaceskip=\fontdimen2\font plus
			\BIBentryALTinterwordstretchfactor\fontdimen3\font minus
			\fontdimen4\font\relax}
		\providecommand{\BIBforeignlanguage}[2]{{%
				\expandafter\ifx\csname l@#1\endcsname\relax
				\typeout{** WARNING: IEEEtran.bst: No hyphenation pattern has been}%
				\typeout{** loaded for the language `#1'. Using the pattern for}%
				\typeout{** the default language instead.}%
				\else
				\language=\csname l@#1\endcsname
				\fi
				#2}}
		\providecommand{\BIBdecl}{\relax}
		\BIBdecl
		
		\bibitem{Hoc}
		A.~Hocquenghem, ``Codes {{Correcteurs}} d'{{Erreurs}},'' \emph{Chiffres},
		vol.~2, pp. 147--156, 1959.
		
		\bibitem{BCHa}
		R.~C. Bose and D.~K. {Ray-Chaudhuri}, ``On a {{Class}} of {{Error Correcting
				Binary Group Codes}},'' \emph{Information and Control}, vol.~3, no.~1, pp.
		68--79, Mar. 1960.
		
		\bibitem{BCHb}
		------, ``Further {{Results}} on {{Error Correcting Binary Group Codes}},''
		\emph{Information and Control}, vol.~3, no.~3, pp. 279--290, Sep. 1960.
		
		\bibitem{McWSl}
		F.~J. MacWilliams and N.~J.~A. Sloane, \emph{The {{Theory}} of {{Error
					Correcting Codes}}}.\hskip 1em plus 0.5em minus 0.4em\relax {Elsevier}, 1977.
		
		\bibitem{Boss}
		M.~Bossert, \emph{Channel {{Coding}} for {{Telecommunications}}}.\hskip 1em
		plus 0.5em minus 0.4em\relax {John Wiley \& Sons, Inc.}, 1999.
		
		\bibitem{Blahut}
		R.~E. Blahut, \emph{Theory and {{Practice}} of {{Error Control Codes}}}.\hskip
		1em plus 0.5em minus 0.4em\relax {Addison-Wesley}, 1983.
		
		\bibitem{guruswami}
		V.~Guruswami and M.~Sudan, ``Improved {{Decoding}} of {{Reed-Solomon}} and
		{{Algebraic-Geometry Codes}},'' \emph{IEEE Transactions on Information
			Theory}, vol.~45, no.~6, pp. 1757--1767, Sep. 1999.
		
		\bibitem{Kolesnik}
		V.~D. Kolesnik and E.~T. Mironchikov, ``Cyclic {{Reed}}\textendash{{Muller
				Codes}} and {{Their Decoding}},'' \emph{Problemy Peredachi Informatsii},
		vol.~4, no.~4, pp. 20--25, 1968.
		
		\bibitem{Kasami}
		T.~Kasami, S.~Lin, and W.~Peterson, ``New {{Generalizations}} of the
		{{Reed-Muller Codes}}\textendash{{I}}: {{Primitive Codes}},'' \emph{IEEE
			Transactions on Information Theory}, vol.~14, no.~2, pp. 189--199, 1968.
		
		\bibitem{infoset}
		M.~Bossert, R.~Schulz, and S.~Bitzer, ``On {{Information Set Decoding}} of
		{{BCH Codes}} over {{Binary Symmetric Channel}},'' {Guangzhou, China}, Dec.
		2020.
		
		\bibitem{Plotkin}
		M.~Plotkin, ``Binary {{Codes With Specified Minimum Distance}},'' \emph{IRE
			Transactions on Information Theory}, vol.~6, no.~4, pp. 445--450, 1960.
		
		\bibitem{Schnabl}
		G.~Schnabl and M.~Bossert, ``Soft-{{Decision Decoding}} of {{Reed-Muller
				Codes}} as {{Generalized Multiple Concatenated Codes}},'' \emph{IEEE
			transactions on information theory}, vol.~41, no.~1, pp. 304--308, 1995.
		
		\bibitem{RmPolar_Mondelli}
		M.~Mondelli, S.~H. Hassani, and R.~L. Urbanke, ``From {{Polar}} to
		{{Reed-Muller Codes}}: {{A Technique}} to {{Improve}} the {{Finite-Length
				Performance}},'' \emph{IEEE Transactions on Communications}, vol.~62, no.~9,
		pp. 3084--3091, 2014.
		
		\bibitem{RmPolar_Arikan}
		E.~Arikan, ``A {{Survey}} of {{Reed-Muller Codes}} from {{Polar Coding
				Perspective}},'' in \emph{2010 {{IEEE Information Theory Workshop}} on
			{{Information Theory}} ({{ITW}} 2010, {{Cairo}})}.\hskip 1em plus 0.5em minus
		0.4em\relax {IEEE}, 2010, pp. 1--5.
		
		\bibitem{stolte}
		N.~Stolte, ``Rekursive {{Codes}} mit der {{Plotkin-Konstruktion}} und ihre
		{{Decodierung}},'' Ph.D. dissertation, Technische Universit\"at Darmstadt,
		2002.
		
		\bibitem{ferdi}
		M.~Bossert and F.~Hergert, ``Hard- and {{Soft-Decision Decoding Beyond}} the
		{{Half Minimum Distance}}\textemdash{{An Algorithm}} for {{Linear Codes}},''
		\emph{IEEE transactions on information theory}, vol.~32, no.~5, pp. 709--714,
		1986.
		
		\bibitem{haeger}
		E.~Santi, C.~H{\"a}ger, and H.~D. Pfister, ``Decoding {{Reed-Muller Codes Using
				Minimum-Weight Parity Checks}},'' in \emph{2018 {{IEEE International
					Symposium}} on {{Information Theory}} ({{ISIT}})}.\hskip 1em plus 0.5em minus
		0.4em\relax {IEEE}, 2018, pp. 1296--1300.
		
		\bibitem{jiongyue}
		J.~Xing, M.~Bossert, S.~Bitzer, and L.~Chen, ``Iterative {{Decoding}} of
		{{Non-Binary Cyclic Codes Using Minimum-Weight Dual Codewords}},'' in
		\emph{2020 {{IEEE International Symposium}} on {{Information Theory}}
			({{ISIT}})}.\hskip 1em plus 0.5em minus 0.4em\relax {IEEE}, 2020, pp.
		333--337.
		
		\bibitem{Chase}
		D.~Chase, ``A {{Class}} of algorithms for decoding block codes with channel
		measurement information,'' \emph{IEEE Transactions on Information theory},
		vol.~18, no.~1, pp. 170--182, 1972.
		
		\bibitem{Dorsch}
		B.~Dorsch, ``A {{Decoding Algorithm}} for {{Binary Block Codes}} and {{J-Ary
				Output Channels}},'' \emph{IEEE Transactions on Information Theory}, vol.~20,
		no.~3, pp. 391--394, 1974.
		
		\bibitem{linfos}
		M.~P. Fossorier and S.~Lin, ``Soft-{{Decision Decoding}} of {{Linear Block
				Codes Based}} on {{Ordered Statistics}},'' \emph{IEEE Transactions on
			Information Theory}, vol.~41, no.~5, pp. 1379--1396, 1995.
		
		\bibitem{Wu}
		D.~Wu, Y.~Li, X.~Guo, and Y.~Sun, ``Ordered {{Statistic Decoding}} for {{Short
				Polar Codes}},'' \emph{IEEE Communications Letters}, vol.~20, no.~6, pp.
		1064--1067, 2016.
		
		\bibitem{Khamy}
		M.~{El-Khamy}, H.-P. Lin, and J.~Lee, ``Binary {{Polar Codes}} are {{Optimised
				Codes}} for {{Bitwise Multistage Decoding}},'' \emph{Electronics Letters},
		vol.~52, no.~13, pp. 1130--1132, 2016.
		
		\bibitem{Liva}
		G.~Liva, L.~Gaudio, T.~Ninacs, and T.~Jerkovits, ``Code {{Design}} for {{Short
				Blocks}}: {{A Survey}},'' \emph{arXiv:1610.00873}, Oct. 2016.
		
		\bibitem{BP_ISD}
		A.~Kothiyal, O.~Y. Takeshita, W.~Jin, and M.~Fossorier, ``Iterative
		{{Reliability-Based Decoding}} of {{Linear Block Codes}} with {{Adaptive
				Belief Propagation}},'' \emph{IEEE Communications Letters}, vol.~9, no.~12,
		2005.
		
		\bibitem{SCC2019}
		M.~Bossert, ``An {{Iterative Hard}} and {{Soft Decision Decoding Algorithm}}
		for {{Cyclic Codes}},'' in \emph{{{SCC}} 2019; 12th {{International ITG
					Conference}} on {{Systems}}, {{Communications}} and {{Coding}}}.\hskip 1em
		plus 0.5em minus 0.4em\relax {VDE}, 2019, pp. 1--6.
		
		\bibitem{arxivboss}
		------, ``On {{Decoding Using Codewords}} of the {{Dual Code}},''
		\emph{arXiv:2001.02956}, Jan. 2020.
		
		\bibitem{bias_ISD}
		W.~Jin and M.~P.~C. Fossorier, ``Reliability-{{Based Soft-Decision Decoding
				With Multiple Biases}},'' \emph{IEEE Transactions on Information Theory},
		vol.~53, no.~1, pp. 105--120, Jan. 2007.
		
		\bibitem{ABP}
		J.~Jiang and K.~Narayanan, ``Iterative {{Soft-Input Soft-Output Decoding}} of
		{{Reed}}\textendash{{Solomon Codes}} by {{Adapting}} the {{Parity-Check
				Matrix}},'' \emph{IEEE Transactions on Information Theory}, vol.~52, no.~8,
		pp. 3746--3756, Aug. 2006.
		
		\bibitem{sage}
		{The Sage Developers}, ``{{SageMath}}, the {{Sage Mathematics Software
				System}},'' 2018.
		
		\bibitem{yuan21}
		J.~Yuan, J.~Xing, and L.~Chen, ``Plausibility {{Analysis}} of {{Shift-Sum
				Decoding}} for {{Cyclic Codes}},'' in \emph{2021 {{IEEE International
					Symposium}} on {{Information Theory}} ({{ISIT}})}, Jul. 2021, pp. 652--657.
		
		\bibitem{Gallager}
		R.~Gallager, ``Low-{{Density Parity-Check Codes}},'' \emph{IRE Transactions on
			information theory}, vol.~8, no.~1, pp. 21--28, 1962.
		
		\bibitem{Massey}
		J.~L. Massey, ``Threshold decoding,'' \emph{Massachusetts Institute of
			Technology, Research Laboratory of Electronics}, 1963.
		
		\bibitem{flippattern}
		A.~Valembois and M.~Fossorier, ``A {{Comparison}} between {{Most-Reliable-Basis
				Reprocessing Strategies}},'' \emph{IEICE TRANSACTIONS on Fundamentals of
			Electronics, Communications and Computer Sciences}, vol. E85-A, no.~7, pp.
		1727--1741, Jul. 2002.
		
		\bibitem{wu2008soft}
		Y.~Wu and M.~Fossorier, ``Soft-{{Decision Decoding Using Time}} and {{Memory
				Diversification}},'' in \emph{2008 {{IEEE International Symposium}} on
			{{Information Theory}}}, Jul. 2008, pp. 76--80.
			
		\bibitem{Grassl:codetables}
		M.~Grassl, ``Bounds on the minimum distance of linear codes and quantum codes,'' available online at \url{http://www.codetables.de}, accessed on 2022-03-24.
	\end{thebibliography}
\end{document}